\documentclass[a4paper, 12pt]{article}

\usepackage[sort&compress]{natbib}
\bibpunct{(}{)}{;}{a}{}{,} 

\usepackage{amsthm, amsmath, amssymb, mathrsfs, multirow, url, subfigure}
\usepackage{graphicx} 
\usepackage{verbatim}
\usepackage[linesnumbered,lined,boxed,commentsnumbered]{algorithm2e}
\usepackage{natbib}
\usepackage{ifthen} 
\usepackage{amsfonts}
\usepackage{longtable}
\usepackage[usenames]{color}
\usepackage{fullpage}

\theoremstyle{plain} 
\newtheorem{thm}{Theorem}

\newtheorem{prop}{Proposition}

\theoremstyle{definition}

\theoremstyle{remark}

\newcommand{\RR}{\mathbb{R}}

\newcommand{\E}{\mathsf{E}}

\newcommand{\prob}{\mathsf{P}}

\newcommand{\Qset}{\mathscr{Q}}

\newcommand{\eps}{\varepsilon}

\newcommand{\logit}{\mathrm{logit}}

\newcommand{\trace}{\mathrm{tr}}

\newcommand{\nm}{\mathsf{N}}

\newcommand{\ber}{\mathsf{Ber}}

\newcommand{\invgamma}{\mathsf{IG}}

\newcommand{\oracle}{{\text{\sc or}}}

\title{Variational approximations of empirical Bayes posteriors in high-dimensional linear models}
\author{Yue Yang\footnote{Department of Statistics, North Carolina State University} \; and \; Ryan Martin$^*$}
\date{\today}

\begin{document}

\maketitle 


\begin{abstract}
In high-dimensions, the prior tails can have a significant effect on both posterior computation and asymptotic concentration rates.  To achieve optimal rates while keeping the posterior computations relatively simple, an empirical Bayes approach has recently been proposed, featuring thin-tailed conjugate priors with data-driven centers.  While conjugate priors ease some of the computational burden, Markov chain Monte Carlo methods are still needed, which can be expensive when dimension is high.  In this paper, we develop a variational approximation to the empirical Bayes posterior that is fast to compute and retains the optimal concentration rate properties of the original.  In simulations, our method is shown to have superior performance compared to existing variational approximations in the literature across a wide range of high-dimensional settings.  
\smallskip

\end{abstract}

\section{Introduction}

Consider the standard Gaussian linear regression model
\begin{equation}
\label{eq:reg.model}
y_i = x_i^\top \beta + \sigma \epsilon_i, \quad i=1,\ldots,n, \quad \text{independent}, 
\end{equation}
where $y_i$ is the response variable, $x_i=(x_{i1},\ldots,x_{ip})^\top \in \RR^p$ is a given vector of predictor variables, $\beta \in \RR^p$ is an unknown vector of regression coefficients, $\sigma > 0$ is an unknown scale parameter, and $\epsilon_i \sim \nm(0,1)$ is the random error term.  In matrix form, this can be written succinctly as $y \sim \nm_n(X\beta, \sigma^2 I_n)$, where $y=(y_1,\ldots,y_n)^\top$ is the vector of response variables and $X$ is the $n \times p$ matrix with $x_i$ as its $i^\text{th}$ row, $i=1,\ldots,n$.  
We are particularly interested in high-dimensional cases, where $p\gg n$.  Without assuming some low-dimensional structure in $\beta$, accurate estimation is hopeless.  As is customary in the literature, here we assume that $\beta$ is {\em sparse} in the sense that most of $\beta_i$'s are zero, but, of course, we do not know how many or which ones are zero.  This sparsity assumption aligns with the belief, common in scientific applications, that only a few of the many predictor variables actually affect the mean response.  Estimating the sparse, high-dimensional $\beta$ vector and/or identifying which entries in $\beta$ are non-zero, i.e., {\em variable selection}, are important problems with many different solutions that have been widely studied.  Regularization-based methods, including lasso \citep{tibshirani1996regression}, adaptive lasso \citep{zou2006adaptive}, elastic net \citep{zou2005regularization}, and SCAD \citep{fan2001variable}, impose different penalty functions on $\beta$ to take advantage of the sparsity assumption.  Bayesian methods instead focus on different choice of prior distributions, such as the normal mixture prior adopted in \citet{george1993variable}, spike-and-slab priors used in \citet{ishwaran2005spike} and \citet{castillo2015bayesian}, the continuous horseshoe prior used in \citet{carvalho2010horseshoe} and \citet{polson2012local}, and the empirical or data-driven priors in \citet{martin2017empirical}, \citet{martin2019empiricalpredict}, and \citet{ecap}.  One obvious advantage to the use of Bayesian methods is that they return an entire posterior distribution for $\beta$, from which lots of interesting and useful summaries can be derived.  The price one pays for this, however, is computational.  That is, the posterior distribution is not available in closed-form and, therefore, must be approximated.  The most common approximation is via Markov chain Monte Carlo (MCMC), but this is well known to be both expensive and inaccurate when $n$ and/or $p$ are large.  An alternative to MCMC is the class of {\em variational approximations}, designed specifically for computational efficiency, is the focus of the present paper.  


Roughly, the variational approach proceeds by first identifying a sufficiently rich yet analytically tractable class of distributions and then choosing the member of that class closest to the posterior distribution with respect to some discrepancy measure.  The computational efficiency gain is a result of converting a difficult integration problem into an optimization problem for which fast algorithms are available.  The most common method is coordinate ascent variational inference \citep{blei2003latent}, which uses  coordinate ascent to minimize the Kullback--Leibler divergence between the mean-field variational family and the true posterior distribution.  Furthermore, stochastic variational inference \citep{hoffman2013stochastic}, black box variational inference \citep{ranganath2014black}, and doubly stochastic variational inference \citep{titsias2014doubly} allow for variational approximations to be applied more generally.  \cite{blei2017variational} gives an authoritative review of variational approximations for Bayesian inference.  Besides the computational efficiency of variational methods, there has been recent interest in the asymptotic theory, e.g.,  \citet{wang2019frequentist,yang2020alpha,alquier2020concentration}.

Variational approximations have been developed for the variable selection problem being considered here.  In particular, \citet{carbonetto2012scalable} define a simple-but-effective variational family to approximate the posterior derived from Gaussian spike-and-slab priors and integrate out hyper-parameters using importance sampling; \citet{huang2016variational} focus on similar spike-and-slab model but update hyper-parameters with maximum a posterior estimate and also propose a novel batch-wise algorithm; and \citet{ormerod2017variational} assume $\sigma^2$ has an inverse gamma distribution and they derive a corresponding update equation. While the above three variational methods all consider spike-and-slab prior with mean zero Gaussian slabs, \citet{ray2019variational} focus on a prior with Laplace slabs.  Their motivation is Theorem~2.8 in \citet{castillo2012needles}, i.e., that Gaussian slabs lead to sub-optimal posterior concentration rates, which suggests the use of a prior with heavier-than-Gaussian tails.  Starting with Laplace instead of Gaussian slabs, \cite{ray2019variational} develop corresponding variational approximations and algorithms, and prove that their proposed approximate posterior distribution enjoys some of the same desirable asymptotic concentration properties as the full posterior.  However, like with its MCMC counterpart, the Laplace prior tails create some computational challenges for the variational approximation.  In particular, the update equations are not available in closed-form, so numerical methods are required at each iteration.  


In this paper, following the insights in \citet{martin2017empirical} and \citet{martin.walker.deb}, we consider posterior distributions obtained by Bayesian updating of suitable {\em empirically-centered} Gaussian priors.  The advantage of these empirical priors is two-fold: they enjoy the computational simplicity and efficiency of thin-tailed conjugate priors and have the optimal posterior concentration rates of heavy-tailed priors.  Although prior conjugacy leads to some computational savings, unfortunately, there is still a need for MCMC methods, which can be expensive when $n$ and/or $p$ are large.  Therefore, like \citet{ray2019variational}, our goal here is to develop a fast variational approximation to this empirical Bayes posterior, one that avoids MCMC altogether.  Moreover, this approximation should not sacrifice on the desirable concentration rate properties of the posterior it is approximating.  After a brief review of the empirical prior formulation from \citet{martin2017empirical} and \citet{martin2019empiricalpredict}, in Section~\ref{Model_intro} we present our variational approximation, its corresponding asymptotic theory, and our algorithm for evaluating that approximation.  Numerical comparisons of our proposed variational method with others for high-dimensional regression are presented in Section~\ref{Simulation}, and there we demonstrate that our method has superior performance across a range of settings.  In Section~\ref{orthogonal} we consider the special case of regression with an orthogonal design matrix, where the variational approximation is sufficiently simple that it allows for further asymptotic concentration properties to be demonstrated, namely, selection consistency and valid uncertainty quantification, under suitable conditions.  Some concluding remarks are given in Section~\ref{S:discuss}, and technical details and proofs are collected in four appendices.


\section{High-dimensional regression}
\label{Model_intro}

\subsection{Empirical prior and the corresponding posterior}
\label{SS:posterior}

Here we adopt the empirical prior formulation as presented in \citet{martin2017empirical} and \citet{martin2019empiricalpredict}.  In particular, we decompose the sparse, high-dimensional vector $\beta$ as $(S, \beta_S)$, where $S \subseteq \{1,2,\ldots,p\}$ is the set of non-zero coefficients, called the {\em configuration} of $\beta$, and $\beta_S$ is the $|S|$-vector of non-zero values, with $|S|$ denoting the cardinality of $S$. We first define prior $\pi(S)$ for the configuration $S$ as
\[ \pi(S)=\textstyle \binom{p}{|S|}^{-1} f_n(|S|),\]
where $f_n(s)$ is a prior on the configuration size $|S|$.  A number of different options for $f_n$ are available; see \citet{castillo2015bayesian}.  One is a suitable beta-binomial prior, but here we will focus on   
\begin{equation}
\label{eq:prior.size}
f_n(s)\propto c^{-s}p^{-as}, \quad s=0,1,\ldots,R,
\end{equation}
where $a$ and $c$ are positive constants and $R = \text{rank}(X)$.  From now on, for simplicity and consistency with the majority of the literature in this area, we will assume that $R=n$; but see \citet{abramovich.grinshtein.2010}.  

For the conditional prior for $\beta_S$, given $S$, \citet{castillo2012needles} showed that thin Gaussian tails can lead to sub-optimal posterior concentration rates, which motivated \citet{castillo2015bayesian} to consider a heavier-tailed Laplace prior.  While the optimal posterior concentration rates can be established with the heavy-tailed conditional prior for $\beta_S$, given $S$, there is a price to pay in terms of posterior computation; a result of the Laplace prior being non-conjugate to the normal likelihood.  But the effect of the prior tails can be reduced considerably by allowing the data to inform the prior center.  Indeed, \citet{martin2017empirical} observed that, with an appropriate empirical Gaussian prior, conjugacy and optimal posterior concentration rate properties could be achieved.  Following their idea, we take the conditional prior for $\beta_S$, given $S$, as 
\begin{equation}
\beta_S \mid S,\sigma^2 \sim \pi_n(\beta_S \mid S) :=  \nm(\hat{\beta}_S,\gamma^{-1}\sigma^2(X_S^\top X_S)^{-1}),
\label{prior_emp}
\end{equation}
where $X_S$ is the sub-matrix corresponding to the configuration $S$, $\hat\beta_S = (X_S^\top X_S)^{-1} X_S^\top y$ is the least squares estimator based on design matrix $X_S$, $\sigma^2$ is the error variance, and $\gamma > 0$ is a scalar tuning parameter that controls the prior spread; see, also, \citet{belitser.ghosal.ebuq}.  For the moment, we will treat $\sigma^2$ as fixed---either at its true value or at a plug-in estimator---but see below.   

For the Gaussian linear regression model, with $\sigma^2$ fixed, the likelihood at $\beta \equiv (S,\beta_S)$ is given by $L_n(S,\beta_S) = \exp\{-\tfrac{1}{2\sigma^2} \|y - X_S\beta_S\|^2\}$.  Then \citet{martin2017empirical} propose the following joint posterior distribution for $(S,\beta_S)$, 
\begin{equation}
\label{joint_post}
\pi^n(S,\beta_S) \propto \tilde\pi^n(S, \beta_S) := L_n^\alpha(S,\beta_S) \, \pi_n(\beta_S \mid S) \, \pi(S), 
\end{equation}
where $\tilde \pi^n$ is the unnormalized posterior distribution, and the proportionality constant that goes in to $\pi^n$ is determined by summing/integrating over all $(S,\beta_S)$.  The power $\alpha \in (0,1)$, which can be arbitrarily close to 1, is an extra regularization factor preventing the posterior---that depends on data through both the likelihood and prior---from over-fitting.  An important consequence of the prior conjugacy is that the marginal posterior distribution for the configuration $S$ is available is nearly closed-form:
\[ \pi^n(S) \propto 
 \tilde\pi^n(S) := \pi(S) \, \bigl(\tfrac{\gamma}{\alpha+\gamma}\bigr)^{|S|/2} \exp\bigl\{-\tfrac{\alpha}{2\sigma^2}\|y-\hat y_S\|^2\bigr\}, \]
where $\hat y_S = X_S \hat\beta_S$ is the least squares fitted value based on configuration $S$. The above expression is the driver behind the MCMC algorithm presented in \citet{martin2017empirical} for sampling from the posterior $\pi^n$ for $(S,\beta_S)$.  They also established a number of desirable asymptotic posterior concentration rate results.  These will be used to prove similar results for the variational approximation developed in Section~\ref{SS:var.approx} below, so a brief summary is presented in Appendix~\ref{summarytheory}.  

In applications, fixing $\sigma^2$ at the true value or at a plug-in estimator may not be fully satisfactory, so \citet{martin2019empiricalpredict} proposed the use of a prior distribution.  In particular, they suggested an inverse gamma prior, $\sigma^2\sim \invgamma(a_0,b_0)$, where $a_0$ and $b_0$ are pre-specified shape and scale parameters.  A very similar marginal posterior for $(S,\beta_S)$ can be developed based on this larger model, but we will not need most of this in what follows.  All that will be relevant to our developments is the corresponding marginal posterior distribution for $S$ which, using the same notation as above, is given by 
\begin{equation}
\label{marginal}
\pi^n(S)\propto \tilde\pi^n(S) := \pi(S) \bigl(\tfrac{\gamma}{\alpha+\gamma}\bigr)^{|S|/2} \bigl(b_0+\tfrac{\alpha}{2} \|y-\hat{y}_S\|^2\bigr)^{-(a_0+\alpha n/2)}.
\end{equation}
This expression will be used in our algorithm for solving the optimization problem that determines our variational approximation; see Section~\ref{SS:algorithm}.  



\subsection{Variational approximation}
\label{SS:var.approx}


With a slight abuse of notation, instead of treating $S$ as a subset of $\{1,2,\ldots,p\}$, treat it as a binary vector, where $S_j=1$ if $j \in S$ and $S_j=0$ otherwise.  Alternatively, we have $\beta_j \neq 0$ if $S_j=1$ and $\beta_j = 0$ otherwise.  One factor that makes computation of the original posterior $\pi^n$ relatively difficult is that the pairs $(S_j,\beta_j)$, $j=1,\ldots,p$, are not independent.  However, for a quick and simple variational approximation, we propose to ignore this dependence and work with a parametric family of the form
\[ q_\theta(S,\beta) = \prod_{j=1}^p q_{j,\theta}(\beta_j \mid S_j) \, q_{j,\theta}(S_j), \]
where $\theta$ is a finite-dimensional parameter to be chosen, which assumes independence across $j$.  Specifically, we take 
\[ q_{j,\theta}(S_j) = \begin{cases} \phi_j & \text{if $S_j=1$} \\ 1-\phi_j & \text{if $S_j=0$} \end{cases} \]
and 
\[ q_{j,\theta}(\beta_j \mid S_j) = \begin{cases} \nm(\beta_j \mid \mu_j, \tau_j^2) & \text{if $S_j=1$} \\ \delta_0(\beta_j) & \text{if $S_j=0$}, \end{cases} \]
where $\delta_0$ denotes the point mass distribution at the origin and $(\mu_j, \tau_j^2, \phi_j) \in \RR \times (0,\infty) \times [0,1]$ are the entries that make up the parameter $\theta$.  If we consider the marginal posterior distribution for $\beta$, the proposed approximation treats each $\beta_j$ as independent and with a mixture of a Gaussian and a point mass distribution, i.e., $\beta_j \sim \phi_j \nm(\mu_j, \tau_j^2) + (1-\phi_j) \delta_0$.  Collecting all such distributions in
\begin{equation}
\label{MF equation}
  \Qset = \Big\{\bigotimes_{j=1}^p \{\phi_j \nm(\mu_j,\tau_j^2)+(1-\phi_j)\delta_0\}:\; \mu_j \in \RR,\;\tau_j^2>0,\;\phi_j \in[0,1]\Big\},
\end{equation}
called the {\em mean-field family} \citep[e.g.,][]{blei2017variational}, the goal then is to find the entry in $\Qset$ that best approximates the posterior $\pi^n$ in a certain sense.  Of course, $\Qset$ is a finite-dimensional set, so this involves optimization only with respect to the parameter $\theta$.  Following \citet{blei2003latent, blei2017variational}, we propose to find the value $\theta$ that maximizes
\begin{equation}
\label{ELBO}
K(\theta) = \E_{(S,\beta) \sim q_\theta} \log\{ \tilde\pi^n(S,\beta) \, / \, q_\theta(S,\beta) \}, 
\end{equation}
the so-called {\em evidence lower bound} \citep[e.g.,][]{blei2017variational}.  Once $\hat\theta = \arg\max_\theta K(\theta)$ is obtained, the variational approximation is 
\[ q^n = q_{\hat\theta}. \]
Solving this optimization problem is not entirely  straightforward, and our proposed computational algorithm will be described in Section~\ref{SS:algorithm} below.




\subsection{Asymptotic theory}
\label{theory}

Here we explore the asymptotic properties of the variational approximation $q^n$ of the empirical Bayes posterior $\pi^n$ in the case of known error variance $\sigma^2$.  To fix ideas and notation, let $\beta^\star$ denote the true $p$-dimensional coefficient vector, where $p$ is possibly much larger than the sample size $n$, in a sense that will be made more precise below.  The $\beta^\star$ vector is sparse in the sense that its configuration $S_{\beta^\star}=\{j: \beta_j^\star \neq 0\}$ is of size relatively small compared to $n$.  In particular, the sample size $n$, the dimension $p$, and the ``effective dimension'' are assumed to satisfy 
\begin{equation}
\label{eq:nps}
|S_{\beta^\star}| = o(n) \quad \text{and} \quad |S_{\beta^\star}| \log(p / |S_{\beta^\star}|) < n, \quad n \to \infty. 
\end{equation}
The latter condition defines the so-called {\em ordinary high-dimensional} setting described in \citet{verzelen2012}.  But note that this allows a wide range of cases, including some where $\log p$ is some power of $n$.   Throughout, the $n \times p$ design matrix $X$ is assumed to be non-random and of rank $n$.  Moreover, we require two additional conditions on $X$.  First, 
\begin{equation}
\label{eq:fullrank}
\text{$X_S^\top X_S$ is non-singular for all $S$ with $|S| \leq n$}
\end{equation}
This is implied by, e.g., the sparse Riesz condition in \citet{zhang2008sparsity}.  It is possible to relax this condition by adjusting the prior distribution to only assign mass to those $S$ such that $X_S^\top X_S$ is non-singular, but this adds considerable complexity.  A typical assumption is that $X$ has rows filled with independent samples from a $p$-variate distribution, e.g., $\nm_p(0,\Psi)$, where $\Psi$ is positive definite, so non-singularity of small sub-matrices is not inconvenient.  Second, to ensure that the special sub-matrix $X_{S^\star}$, with $S^\star = S_{\beta^\star}$, is sufficiently stable, we require that 
\begin{equation}
\label{eq:eigen}
\lambda_{\min}(n^{-1} X_{S^\star}^\top X_{S^\star}) \gtrsim |S^\star| p^{-1}, \quad S^\star = S_{\beta^\star}, 
\end{equation}
where $\lambda_{\min}$ denotes the minimal eigenvalue operator.  Since $|S_{\beta^\star}|$ is small according to \eqref{eq:nps}, and since $|S_{\beta^\star}| p^{-1}$ is vanishing rapidly, this too is not a major restriction.  Finally, we write ``$\prob_{\beta^\star}$'' and ``$\E_{\beta^\star}$'' below to indicate probability and expectation with respect to the distribution of $y$ in \eqref{eq:reg.model} when $\beta^\star$ is the true coefficient vector. 

Below we present three results pertaining to the asymptotic concentration of $q^n$; proofs of all three are given in Appendix~\ref{App:proof}.  Our arguments are based on the beautiful result in \citet[][Theorem~7]{ray2019variational} that connects the concentration properties of the variational approximation to those of the posterior being approximated, and a bound on the Kullback--Leibler divergence between the two.  Various concentration rate results for $\pi^n$ have been established elsewhere (see Appendix~\ref{summarytheory}), so we only need to check this latter condition on the Kullback--Leibler divergence.  It turns out that our derivations are much simpler here, thanks to the conjugate normal prior, compared to the Laplace prior formulation in \citet{ray2019variational}.  

The first result reveals that the data-dependent distribution $q^n$ for the vector $\beta$, concentrates near $\beta^\star$ in the sense that the event ``$\|X(\beta-\beta^\star)\|_2$ is relatively large'' has vanishing $q^n$-probability.  Moreover, the concentration rate---the precise notion of ``relatively large''---is minimax optimal.  Indeed, define  
\begin{equation}
\label{eq:rate}
\eps_n^2(\beta^\star) = |S_{\beta^\star}| \log(p / |S_{\beta^\star}|)
\end{equation}
which, as Theorem~\ref{prediction_error} shows, determines the $q^n$ asymptotic concentration rate.  




\begin{thm}
\label{prediction_error}
Under the setup described above, with \eqref{eq:nps}, \eqref{eq:fullrank}, \eqref{eq:eigen}, $\eps_n^2(\beta^\star)$ as in \eqref{eq:rate}, and any sequence $M_n$ with $M_n \to \infty$, the variational approximation $q^n$ satisfies
\[ \sup_{\beta^\star} \E_{\beta^\star} q^n(\{\beta \in \RR^p: \|X(\beta-\beta^\star)\|_2^2 > M_n \eps_n^2(\beta^\star)\}) \to 0, \quad n \to \infty, \]
where the supremum is over all $\beta^\star$ with $|S_{\beta^\star}|=o(n)$. 
\end{thm}

The second result concerns the ``effective dimension'' of $q^n$.  Although $q^n$ is a distribution supported on all of $\RR^p$, having point mass mixture components implies that some of those $p$ dimensions are effectively collapsed, thereby reducing the effective dimension.  The following theorem establishes that the effective dimension of $q^n$ is not too much larger than the effective dimension $|S_{\beta^\star}|$ of the true $\beta^\star$.  

\begin{thm}
\label{effective dimension}
Under the setup of Theorem~\ref{prediction_error}, for any sequence $M_n > 1$ with $M_n \to \infty$, the variational approximation $q^n$ satisfies 
\[ \sup_{\beta^\star}\E_{\beta^\star} q^n(\{\beta\in \RR^p:|S_\beta|>M_n |S_{\beta^\star}|\}) \to 0, \quad n \to \infty,\]
where the supremum is over all $\beta^\star$ with $|S_{\beta^\star}|=o(n)$. 
\end{thm}

The third and final result of this section concerns the direct concentration of $q^n$ around $\beta^\star$, i.e., in terms of the distance between $\beta$ and $\beta^\star$ rather than distance between the corresponding mean responses.  For this, some extra conditions on the $X$ matrix are required, roughly, to ensure that certain sub-matrices---based on subsets of the columns of $X$---are full rank.  In particular, define the smallest scaled sparse singular value of $X$ of dimension $s$ as
\begin{equation}
\label{eq:kappa}
\kappa_X(s)=\inf_{\beta:0<|S_\beta|\leq s}\frac{\|X\beta\|_2}{\|\beta\|_2},\quad s=1,\dots,p.
\end{equation}
\citet{ariascastro.lounici.2014} show that a sparse, high-dimensional $\beta^\star$ is identifiable from a model with design matrix $X$ if and only if $\kappa_X(2|S_{\beta^\star}|) > 0$.  Slightly more than identifiability is needed here---and in all other papers on this topic---to establish concentration rates with respect to $\|\beta-\beta^\star\|_2$ and, we assume that $\kappa_X(C|S_{\beta^\star}|) > 0$ for a suitable constant $C > 2$.  This is implied by \eqref{eq:fullrank}, the difference here is that the concentration rate with respect to $\|\beta-\beta^\star\|$ is determined by how fast $\kappa_X(C|S_{\beta^\star}|)$ approaches 0.  For more on identifiability, see, e.g.,  \citet{ariascastro.lounici.2014} and \citet{castillo2015bayesian}.

\begin{thm}
\label{estimation_error}
Under the setup in Theorem~\ref{prediction_error}, in particular, with \eqref{eq:nps}, \eqref{eq:fullrank}, and \eqref{eq:eigen}, for any sequence $M_n$ such that $M_n \to \infty$, the variational approximation $q^n$ satisfies
\[ \E_{\beta^\star} q^n\Bigl( \Bigl\{\beta\in \RR^p: \|\beta-\beta^\star\|^2_2 > \frac{M_n \eps_n^2(\beta^\star)}{\kappa_X^2(C|S_{\beta^\star}|)} \Bigr\}\Bigr) \to 0, \quad n \to \infty,\]
for all $\beta^\star$ such that $|S_{\beta^\star}|=o(n)$ for $\kappa_X$ in \eqref{eq:kappa} and $C > 2$.  
\end{thm}

\subsection{Algorithm}
\label{SS:algorithm}

In existing work on variational approximations in linear regression settings, typically the prior distribution treats the entries of $\beta$ as independent of even independent and identically distributed (iid).  For example, in \citet{ray2019variational}, the entries of $\beta$ are {\em a priori} iid, with the marginal prior for $\beta_j$ a mixture of a point mass at 0 and a Laplace distribution centered at 0, $j=1,\ldots,p$.  It is not possible, however, to recast our empirical prior as iid, so a different approach is needed.  

The key observation is that, for any fixed error variance $\sigma^2$, thanks to the conjugate normal form of the empirical prior/posterior, it is possible to derive explicit expressions for coordinate ascent updates to the parameter $\theta=\{(\mu_j, \tau_j^2, \phi_j):j=1,\ldots,p\}$ in the variational family.  After standardizing $X$ and $y$, so that
\begin{equation}
\label{eq:standardize}
\textstyle \sum_{i=1}^n y_i=0,\quad \sum_{i=1}^n x_{ij}=0,\quad \text{and} \quad \sum_{i=1}^n x_{ij}^2=n, 
\end{equation}
we maximize the evidence lower bound in \eqref{ELBO} via coordinate ascent, for a fixed $\sigma^2$, via the updates
\begin{align}
\mu_j^{(t+1)} & =\frac{(X^\top y)_j-\sum_{k< j}(X^\top X)_{jk}\phi_k^{(t+1)} \mu_k^{(t+1)}-\sum_{k> j}(X^\top X)_{jk}\phi_k^{(t)} \mu_k^{(t)}+\frac{\gamma g(\tilde S)}{\alpha}\tilde\beta_j}{n+\gamma g(\tilde S)/\alpha} \notag \\ 
\tau_j^{2(t+1)} & = \frac{\sigma^2}{n(\alpha+\gamma)} \label{eqn:update_equ} \\
\logit\,\phi_j^{(t+1)} & = \frac12 \log\frac{\gamma g(\tilde S) }{n(\alpha+\gamma)} 
+ \Bigl\{ \frac{n\alpha}{2} + \gamma g(\tilde S) \Bigr\} \frac{\mu_j^{(t+1)2}}{\sigma^2} -\frac{\gamma g(\tilde S)}{2\sigma^2}(\mu_j^{(t+1)} - \tilde\beta_j)^2 \notag \\
& \qquad -\log c-a\log p, \notag 
\end{align}
where $g(S)$ denotes the geometric mean of the eigenvalues of $X_S^\top X_S$, $\tilde\beta$ is the lasso estimator, and $\tilde{S} = \{j: \tilde\beta_j \neq 0\}$ is the configuration selected by lasso.  The algorithm is stopped if for each $j$, the change in entropy between $\ber(\phi_j^{(t+1)})$ and $\ber(\phi_j^{(t)})$ is less than a prespecified threshold $\delta$, that is, we stop the iterations when 
\[ \max_j \bigl|H(\phi^{(t+1)}_j)-H(\phi^{(t)}_j)\bigr|<\delta, \]
where $H(\phi_i)= -p\log_2 p-(1-p)\log_2 (1-p)$.  Detailed derivations of the update equations in \eqref{eqn:update_equ} are presented in Appendix~\ref{derivation}.  Note that certain steps of these derivations make some simplifying assumptions about the $X_S^\top X_S$ matrix.  In particular, after standardizing the columns of $X$ as in \eqref{eq:standardize} and the making the full-rank assumption in \eqref{eq:fullrank}, it is not unreasonable to expect a certain  ``homogeneity'' in $X_S^\top X_S$ as $S$ varies.  That is, the spectrum of $X_S^\top X_S$ should be relatively narrow and relatively insensitive to changes in $S$.  This boils down to effectively ignoring the off-diagonal terms in $X_S^\top X_S$, which is what the variational approximation proposes to do anyway.  


Having explicit update equations is an advantage, but these are not immediately applicable because, of course, the error variance $\sigma^2$ is unknown in practice.  One obvious work-around is to replace $\sigma^2$ with a plug-in estimator $\hat\sigma^2$, e.g., \citet{ray2019variational} use the lasso-based estimator in \citet{reid.tibshirani.friedman.2014}, implemented in the {\tt selectiveInference} package in R. Alternatively, \citet{huang2016variational} update $\sigma^2$ with the maximum a posteriori  estimate at each iteration of coordinate ascent.  We found that these two strategies, combined with our update equations presented above, led to rather unstable performance in simulations.  Therefore, we opt for a modified version of the importance sampling-based procedure used in \citet{carbonetto2012scalable}.  In particular, specify a range of $\sigma^2$ values, denoted by $\Sigma = \{\varsigma_1^2,\ldots,\varsigma_L^2\}$, and, for each $\ell=1,\ldots,L$, apply the aforementioned coordinate ascent procedure to get parameter estimates 
\[ \theta(\ell) = \bigl\{ (\mu_j(\ell), \tau_j^2(\ell), \phi_j(\ell)): j=1,\ldots,p\}, \quad \ell=1,\ldots,L, \]
each based on treating $\sigma^2 = \varsigma_\ell^2$ as fixed.  In addition, define 
\[ \hat S(\ell) = \{j: \phi_j(\ell) > \tfrac12\} \]
as the selected configuration based on the $\ell^\text{th}$ fixed variance, and evaluate the weights  
\begin{equation}
\label{eq:S.weight}
\tilde w_\ell = \tilde\pi^n(\hat S(\ell)), \quad \ell=1,\ldots,L, 
\end{equation}
where $\tilde\pi^n$ is the unnormalized marginal posterior for the configuration $S$ in \eqref{marginal} based on the formulation with an inverse gamma prior for $\sigma^2$.  Finally, we summarize the $L$ different variational family parameter estimates as 
\begin{equation}
\label{eq:updates.avg}
\textstyle \mu_j = \sum_{\ell=1}^L w_\ell \mu_j(\ell), \quad \tau_j^2 = \sum_{\ell=1}^L w_\ell \tau_j^2(\ell), \quad \text{and} \quad \phi_j = \sum_{\ell=1}^L w_\ell \phi_j(\ell),
\end{equation}
where $w_\ell = \tilde w_\ell / \sum_{\ell=1}^L \tilde w_\ell$ are the normalized weights.  Our approach is similar to that in the {\tt varbvs} procedure \citep{carbonetto2012scalable} in the sense that both use a weighted average of various fixed-$\sigma^2$ parameter estimates.  However, our update mechanism is different from theirs in two important aspects.  First, {\tt varbvs} is approximating  integration over a three-dimensional hyperparameter space with importance sampling, which requires hundreds of samples, while our approach resembles a grid search on $\Sigma$ which requires less than 10 samples.  Second, {\tt varbvs} calculates importance weights based on the evidence lower bound while we use the marginal posterior probability \eqref{marginal} evaluated at a selected configuration instead, which we found to have superior empirical performance compared to other techniques. A summary of our proposed procedure, {\em VB-empirical}, is presented in Algorithm~\ref{algo}.

It is worth noting that \citet{huang2016variational} propose a batch-wise coordinate ascent algorithm where they update the entire $\mu$, $\tau^2$, or $\phi$ vector all at once instead of one entry at a time.  Although a version of this batch-wise algorithm could be easily derived in our context, we found the results to be relatively unstable compared to a standard one-at-a-time update. Also,  \citet{ray2019variational} notice that the standard one-at-a-time updates are sensitive to the ordering of the parameters and, therefore, they propose a prioritized updating scheme. In particular, variables are updated in decreasing order according to an initial estimate $\mu^{(0)}$, so that important variables are expected to be updated first. We also employ this simple yet efficient strategy in our algorithm to avoid the sensitivity to updating order.   


\begin{algorithm}[t]
\SetAlgoLined
\KwIn{standardized data $(X,y)$; a fixed estimator $\tilde\beta$ based on, say, lasso; a grid $\Sigma = \{\varsigma_1^2,\ldots,\varsigma_L^2\}$ of error variances; and a stopping threshold $\delta$.}

Initialize $\theta=(\mu,\tau^2,\phi)$ and set $d=\text{order}(|\mu|)$.

\For{$\ell$ in $1,\ldots,L$}{

$\sigma^2= \varsigma_\ell^2$

\Repeat{$\max_j |H(\phi_j')-H(\phi_j)|<\delta$}{
$\phi'=\phi$;

\For{$k$ in $1,\ldots,p$}{
    
    $j=d_k$;
    
    update $\mu_j$, $\tau_j^2$, and $\phi_j$ according to \eqref{eqn:update_equ};
    

}
}
\textbf{return}
$\mu(\ell)$, $\tau^2(\ell)$, and $\phi(\ell)$; and $\tilde w_\ell = \pi^n(\hat S(\ell))$ as in \eqref{eq:S.weight} 

}

\textbf{output} weighted averages $\mu$, $\tau^2$, and $\phi$ as in \eqref{eq:updates.avg}.

\caption{VB-empirical --- variational approximation for empirical Bayes}
\label{algo}
\end{algorithm}

\section{Numerical comparisons}
\label{Simulation}

\subsection{Methods}

In this section, we compare three variational methods: varbvs from \citet{carbonetto2012scalable}, VB-Gaussian from \citet{huang2016variational} and VB-Laplace from \citet{ray2019variational}
with our VB-empirical in different scenarios. Also, we include the results from Lasso as a benchmark.  For a fair comparison, we let all variational methods start from the lasso estimator and set the same stopping criteria, i.e., the convergence is determinated by the maximum entropy difference as defined in Algorithm~\ref{algo} and define $\delta=10^{-4}$.  For varbvs, we use the R package ``varbvs" and set all parameters as default.  For VB-Gaussian, we use the component-wise VB in \citet{huang2016variational} instead of the batch-wise version since we did not find significant improvement using the latter in our simulations. We further take $v_1=100$, $v=1$, $\lambda=1$, $a_0=1$ and $b_0=p$ in VB-Gaussian. For VB-Laplace, we set $\lambda=1$, $a_0=1$, $b_0=p$ and estimate the regression error term $\sigma^2$ using R package ``selectiveInference".  For VB-empirical, we let $c=1$, $a=0.05$, $\alpha=0.99$, $\gamma=0.005$, and set the initial $\beta$ estimate based on lasso.  For the candidate variance set $\Sigma$, we first find an estimation $\hat\sigma^2$ through ``selectiveInference". Centered at $\hat{\sigma}^2$, we define a interval $[\hat{\sigma}^2/5,9\hat{\sigma}^2/5]$ and choose $L=10$ equally spaced values in this interval as $\varsigma_1^2,\dots,\varsigma_L^2$. It is also worth noted that while \citet{ray2019variational} observe that the VB-Laplace's performance is sensitive to the update ordering, it is a common problem for all variational methods. Fortunately, this sensitivity could be resolved by the prioritized updating scheme proposed in \citet{ray2019variational}, and this technique could be easily accommodated by other variational methods. 

For each scenario, 100 data sets are randomly generated. The design matrix is generated from a multivariate normal distribution, with mean 0 and unit marginal variances.  We compare these methods based on four metrics: averaged $\ell_2$ estimation error, averaged model size, proportion of correct model identifications, $\prob(\hat S = S^\star)$, and the proportion of correct model inclusion, $\prob(\hat S \supseteq S^\star)$. For each variational method, we return $\hat{\beta}_j=\phi_j \mu_j$ for estimation and $\hat S = \{j: \phi_j > \frac12\}$ for model selection.  

\subsection{Simulation I: effect of dimension}
In this section, we compare five methods' performance under different combinations of $n$, $p$, $s$, where $s$ is the number of important ones. we consider five different cases as follows.
\begin{enumerate}
\item $n=100,\;p=400,\;s=10,\;\beta_{S^\star}^\star=(0.5,1.0,1.5,\dots,4.5,5.0,0,\dots,0)^\top.$

\item $n=200,\;p=400,\;s=10,\;\beta_{S^\star}^\star=(0.5,1.0,1.5,\dots,4.5,5.0,0,\dots,0)^\top.$

\item 
$n=100,\;p=400,\;s=20$,
\[ \beta_{S^\star}^\star=(\text{rep}(0.5,5),\text{rep}(1,5),\text{rep}(1.5,5),\text{rep}(2.0,5),0,\dots,0)^\top \]

\item 
$n=200,\;p=800,\;s=20,\;\beta_{S^\star}^\star=(0.5,1.0,1.5,\dots,9.5,10.0,\dots,0)^\top.$

\item
$n=200,\;p=1600,\;s=40,\;\beta_{S^*}=(\beta^*,0,\dots,0)^\top,$where $\beta^\star$ is a sequence with 40 equally spaced values from 1 to 10.
\end{enumerate}
For all five scenarios, the design matrix $X$ is generated from a multivariate normal distribution with mean zero and identity covariance matrix. 

The simulation results are shown in Table~\ref{table:sim_dimension}. VB-empirical and varbvs perform significantly better than the other three across five different dimension settings and Lasso tends to choose large models in all settings.  For case~1 and case~2, all variational methods have similarly good performance. Case 3 is more challenging due to the first five small signals. In this case, VB-empirical performs the best in both estimation error and model selection performance. The performance of varbvs is similar to that of VB while VB-Gauss and VB-Laplace tend to ignore some small signals in this case.  For case 4 and 5 where dimensions are higher, varbvs and VB-empirical significantly outperform the other three. Comparing varbvs and VB-empirical, the latter's performance is more stable than the former's since the standard error of $\ell_2$ estimation for VB-Emp is only 0.06 in Case~5 and that for varbvs is 3.76 even though the difference between the averaged $\ell_2$ estimations is small.  It is also worth noting that, in Case~5, even though lasso choose vary large model, it does not include all important ones and VB-empirical could still identify the true model with incorrect lasso estimation.

\begin{table}[t]
\begin{center}
\begin{tabular}{c c c c c c}
 \hline
 Case & Method&$\E\|\hat\beta-\beta^\star\|^2$ (SE)  & $\E|\hat S|$&$\prob(\hat{S}\supseteq S^\star)$&$\prob(\hat{S}=S^\star)$\\
 \hline
  1&Lasso&1.10(0.18)&19.44&0.85&0 \\
   &VB-Gauss&0.40(0.13)&9.99&0.75&0.60\\
   &VB-Laplace&0.46(0.17)&10.38&0.69&0.49\\
    &varbvs&0.41(0.14)&9.80&0.75&0.72\\  
    &VB-Emp&0.43(0.11)&9.84&0.73&0.66\\
\hline
 2&Lasso&0.68(0.10)&16.50&1&0.03 \\
   &VB-Gauss&0.23(0.07)&10.01&0.98&0.95\\
   &VB-Laplace&0.24(0.08)&10.10&0.98&0.91\\
    &varbvs&0.23(0.07)&10.03&1&0.97\\  
    &VB-Emp&0.27(0.06)&10.04&0.99&0.95\\
 \hline
 3&Lasso&2.18(0.61)&35.53&0.15&0 \\
   &VB-Gauss&1.20(0.96)&16.80&0.16&0.14\\
   &VB-Laplace&1.23(0.49)&16.91&0.05&0.01\\
    &varbvs&0.89(0.42)&18.56&0.27&0.18\\  
    &VB-Emp&0.82(0.23)&19.01&0.33&0.19\\
 \hline
 4&Lasso&1.20(0.16)&32.69&0.98&0 \\
   &VB-Gauss&0.47(0.14)&28.2&0.94&0.04\\
   &VB-Laplace&0.56(0.16)&34.21&0.91&0.04\\
    &varbvs&0.35(0.08)&19.98&0.97&0.96\\  
    &VB-Emp&0.39(0.09)&19.92&0.95&0.95\\
 \hline
 5&Lasso&10.90(4.60)&76.89&0.13&0 \\
   &VB-Gauss&3.59(3.38)&39.16&0.25&0.06\\
   &VB-Laplace&3.14(2.45)&43.10&0.13&0\\
    &varbvs&0.88(3.76)&39.67&0.99&0.96\\  
    &VB-Emp&0.52(0.06)&40&1&1\\
 \hline
\end{tabular}
\end{center}
\caption{Results for Simulation I.}
\label{table:sim_dimension}
\end{table}

\subsection{Simulation II: effect of signal size}
In this section, we consider different signal sizes under fixed dimensions. In particular, we fix $n=200,\;p=1600,\;s=40$ and run simulations in large, moderate and small signals cases as follows.
\begin{enumerate}
    \item $\beta_{S^\star}^\star=(10,\dots,10,0,\dots,0)^\top.$
    \item $\beta_{S^\star}^\star=(1,\dots,1,0,\dots,0)^\top.$
    \item $\beta_{S^\star}^\star=(0.6,\dots,0.6,0,\dots,0)^\top.$
\end{enumerate}
In all three considered settings, VB-empirical performs significantly better than all other methods. While variable selection problem in this dimension setting is difficult, VB-empirical has strong performance in terms of both estimation error and identifying true models.  In contrast, varbvs selects very small models in all cases, especially in large and small signal cases, which might be caused by improper choice of samples for $\sigma^2$.  Similarly, in these cases where lasso chooses many incorrect predictors, the variance estimation based on lasso is unreliable and affects the performance of VB-Laplace as well.

\begin{table}[t]
\begin{center}
\begin{tabular}{c c c c c c}
 \hline
 $\beta_i^*$ & Method&$\E\|\hat\beta-\beta^\star\|^2$ (SE)  & $\E|\hat S|$&$\prob(\hat{S}\supseteq S^\star)$&$\prob(\hat{S}=S^\star)$\\
 \hline
  10&Lasso&38.65(14.58)&68.15&0.29&0 \\
   &VB-Gauss&26.36(28.51)&27.75&0.51&0.09\\
   &VB-Laplace&27.85(24.65)&55.29&0.48&0.03\\
    &varbvs&58.56(10.63)&5.90&0.03&0.03\\  
    &VB-Emp&0.51(0.06)&40&1&1\\
\hline
 1&Lasso&4.72(0.80)&61.35&0.07&0  \\
   &VB-Gauss&4.84(2.07)&13.83&0.14&0.11\\
   &VB-Laplace&4.54(2.06)&26.43&0.16&0.10\\
    &varbvs&5.14(2.03)&11.71&0.16&0.08\\  
    &VB-Emp&0.53(0.08)&40.36&1&0.82\\
 \hline
 0.6&Lasso&3.14(0.32)&50.60&0&0  \\
   &VB-Gauss&3.61(0.27)&4.90&0&0\\
   &VB-Laplace&3.49(0.44)&18.65&0.01&0\\
    &varbvs&3.68(0.31)&3.38&0.01&0\\  
    &VB-Emp&1.84(1.43)&26.72&0.58&0.20\\
 \hline
\end{tabular}
\end{center}
\label{table:sim}
\caption{Results for Simulation II. $\beta_i^*$ represents the value of important signal.}
\end{table}

\subsection{Simulation III: effect of correlation}
In this section, we fix
\[n=100,p=400,s=10,\]
\[\beta=(0.6,0.9,1.2,1.5,1.8,2.1,2.4,2.7,3.0,3.3,0,\dots,0)^\top.\]
The design matrix is generated from a multivariant normal distribution with mean 0 and covariance matrix having element $(i,j)$ being $\rho^{|i-j|}$. We vary $\rho$ among 0.2, 0.5 and 0.8 to explore how different predictor correlations would affect simulation results. The simulation results are recorded in Table~\ref{table:correlation}. 

The general conclusion is the same as that in all other cases, that is, VB-Emp perfroms most stably in all cases. VB-Gauss and VB-Laplace tend to choose larger models with the increase of collinearity.  However, VB-Laplace does not always identify all important ones even though it choose many unimportant variables, especially in high correlation cases. In contrast, varbvs and VB-Emp prefer small models in high dimensional cases and might thus ignore some important variables. But we argue that according to the $l_2$ estimation error and the averaged model size in Case 3, it's most likely that VB-Emp only ignore one predictor with small signal, which is still acceptable in practice.

\begin{table}[t]
\begin{center}
\begin{tabular}{c c c c c c}
 \hline
 $\rho$ & Method& $\E\|\hat\beta-\beta^\star\|^2$ (SE)  & $\E|\hat S|$&$\prob(\hat{S}\supseteq S^\star)$&$\prob(\hat{S}=S^\star)$\\
 \hline
  $0.2$&Lasso&0.82 (0.13)&16.08&0.99&0.05\\
   &VB-Gauss&0.38 (0.12)&10.15&0.93&0.80\\
   &VB-Laplace&0.47 (0.19)&11.58&0.91&0.48\\
    &varbvs&0.37 (0.12)&10.03&0.94&0.86\\  
    &VB-Emp&0.41 (0.11)&10.07&0.93&0.81\\
\hline
 $0.5$&Lasso&0.65 (0.12)&12.62&0.99&0.26 \\
   &VB-Gauss&0.51 (0.17)&12.58&0.94&0.29\\
   &VB-Laplace&1.01 (0.34)&25.99&0.51&0\\
    &varbvs&0.51 (0.21)&9.82&0.76&0.72\\  
    &VB-Emp&0.53 (0.17)&9.92&0.78&0.72\\
 \hline
 $0.8$&Lasso&0.79 (0.21)&10.70&0.95&0.56  \\
   &VB-Gauss&0.97 (0.48)&23.99&0.97&0.01\\
   &VB-Laplace&3.10 (1.21)&41.94&0.04&0\\
    &varbvs&2.19 (0.96)&8.25&0.0&0\\  
    &VB-Emp&1.06 (0.35)&9.13&0.20&0.20\\
 \hline
\end{tabular}
\end{center}
\caption{Results for Simulation III.}
\label{table:correlation}
\end{table}

\section{Special case: orthogonal design}
\label{orthogonal}

\subsection{Simpler model and approximation}

Consider the case where $p \leq n$ and the design matrix $X$ is orthogonal, i.e., $X^\top X = I_p$.  In such cases, via a simple linear transformation, the original regression problem in \eqref{eq:reg.model} can be recast as a sparse, high-dimensional normal means model.  Indeed, if we set $y \gets X^\top y$ and $n \gets p$, then we have 
\begin{equation}
\label{eq:means.model}
y_i \sim \nm(\beta_i, \sigma^2), \quad i=1,\ldots,n, \quad \text{independent}.
\end{equation}
We continue to assume that the $n$-vector $\beta$ is sparse in the sense that most of its entries are zero.  The goal is to make inference on the sparse $\beta$ vector and, in particular, to identify which entries are non-zero.  Although this is a very special case of the original regression problem, it is interesting in its own right.  Indeed, this model is common in all sorts of signal detection problems from image denoising \citep[e.g.,][]{abramovich2006, donohojohnstone1994a, johnstonesilverman2005} to genomics \citep[e.g.,][]{efron2004, jincai2007, mt-test}.

For this version of the problem, the empirical prior construction can proceed almost the same as before. As a first step, since we are assuming the same kind of sparsity as before, the reparametrization $\beta \equiv (S,\beta_S)$ is appropriate here too and, therefore, so is the hierarchical empirical prior formulation. Following \citet{martin2019empirical}, who build upon the original work in \citet{martin2014asymptotically}, we set the marginal prior for $S$ as 
\[ \pi(S) = \textstyle \binom{n}{|S|}^{-1} f_n(|S|), \]
where $f_n$ is as in \eqref{eq:prior.size}, but with $p \equiv R \equiv n$.  Then the conditional prior for $\beta_S$, given $S$, can be written simply as 
\[ \beta_S \mid S, \sigma^2 \sim \pi_n(\beta_S \mid S) := \nm_{|S|}(y_S, \gamma^{-1} \sigma^2 I_{|S|}). \]
See, also, \citet{belitser.ddm} and \citet{belitser.nurushev.uq}. Combining the joint empirical prior $\pi_n(S,\beta_S) = \pi(S) \pi_n(\beta_S \mid S)$ with the ($\alpha$ power of the) likelihood, yields a posterior $\pi^n(S, \beta_S)$ exactly as before.  Algorithms for posterior sampling along with asymptotic posterior concentration rate results are presented in the aforementioned papers.  Here the goal is to develop an appropriate variational approximation to the posterior distribution $\pi^n$ and investigate its properties.  


For this simpler model, it turns out that we {\em can} rewrite the prior in a simple independent spike-and-slab style.  Indeed, the marginal prior for $\beta_i$ is of the form 
\[ \beta_i \sim \lambda \, \nm(y_i, \sigma^2 \gamma^{-1}) + (1-\lambda) \, \delta_0, \quad i=1,\ldots,n, \]
where $\lambda \equiv \pi(S \ni i)$ is the prior inclusion probability, which does not depend on $i$---provided that the prior for $S$ is uniform on configurations of a given size.  Although $\lambda$ does not depend on an individual $i$, it does depend on the sample size, $n$, so we will henceforth write $\lambda_n$.  In fact, it is relatively easy to show that $\lambda_n = n^{-1} \E|S|$, where the latter is the prior mean for $|S|$ under $f_n$.  If, as before, we let 
\[ f_n(s) \propto (c n^a)^{-s}, \quad s=0,1,\ldots,n, \]
for constants $a,c > 0$, then it can be shown that $\lambda_n = n^{-1} \E|S| = O(n^{-(a+1)})$.  For simplicity, in what follows, we take 
\begin{equation}
\label{eq:lambda}
\lambda_n = n^{-(a+1)}.
\end{equation}

Since the form of our prior distribution matches that of the variational approximation we seek, and the data are independent, it follows that the exact posterior distribution, $\pi^n$, also has that form.  Computation of the full posterior is doable---see \citet{martin2014asymptotically} and \citet{martin2019empirical}---but the posterior inclusion probabilities require MCMC.  It turns out that there is a simple and accurate variational approximation.  If, as in Section~\ref{SS:var.approx}, we work with a mean-field approxiation family of the form 
\[ \bigotimes_{i=1}^n \{ \phi_i \, \nm(\mu_i, \tau_i^2) + (1-\phi_i) \, \delta_0 \}, \]
then the corresponding update equations are 
\begin{align*}
\mu_i & = y_i \\
\tau_i^2 & = \sigma^2 (\alpha + \gamma)^{-1} \\
\logit(\phi_i) & = \logit(\lambda_n) + \tfrac12 \log\tfrac{\gamma}{\alpha+\gamma}+\tfrac{\alpha}{2\sigma^2}y_i^2.
\end{align*}
Of course, these are actually expressions for the estimates, not ``updates,'' and they determine the variational approximation $q^n$.  Having relatively simple expressions for the variational family parameter estimates makes it possible to establish some additional theoretical convergence properties, namely, selection consistency and valid uncertainty quantification; see Section~\ref{SS:theory.orthogonal} below.



\subsection{More asymptotic theory}
\label{SS:theory.orthogonal}

Since this sparse normal means model is a special case of the regression problem considered previously, we can immediately specialize Theorems~\ref{prediction_error}--\ref{effective dimension} to this case.  Analogous to the previous setting, we assume $|S_{\beta^\star}|=o(n)$ and that $|S^\star|\log(n/|S_{\beta^\star}|) < n$.  Then the minimax optimal concentration rate is 
\[ \eps_n^2(\beta^\star) = |S_{\beta^\star}| \log(n / |S_{\beta^\star}|). \]

\begin{thm}
\label{thm:means}
Let $\pi^n$ be the posterior based on the empirical prior described above, and $q^n$ the corresonding variational approximation.  
\begin{enumerate}
\item For any sequence $M_n > 0$ with $M_n \to \infty$, 
\[ \sup_{\beta^\star} \E_{\beta^\star} q^n(\{\beta \in \RR^n: \|\beta-\beta^\star\|_2^2 > M_n \eps_n^2(\beta^\star)\}) \to 0, \quad n \to \infty, \]
where the supremum is over all $\beta^\star$ with $|S_{\beta^\star}|=o(n)$.  
\item For any sequence $M_n > 1$ with $M_n \to \infty$, 
\[ \sup_{\beta^\star}\E_{\beta^\star} q^n(\{\beta\in \RR^p:|S_\beta|>M_n |S_{\beta^\star}|\}) \to 0, \quad n \to \infty,\]
where the supremum is over all $\beta^\star$ with $|S_{\beta^\star}|=o(n)$. 
\end{enumerate}
\end{thm}

The variational approximation is very simple in this setting, so we can establish more than just these basic concentration rate result.  In particular, below we show that the variational approximation will, under some conditions, identify the correct configuration $S_{\beta^\star}$ asymptotically, which implies a variable/model selection consistency property.  Moreover, we also establish that certain marginal distributions derived from the variational approximation achieve valid uncertainty quantification.  

The next result shows that, asymptotically, the variational approximation $q^n$ will not assign positive mass to proper supersets of $S_{\beta^\star}$.  To ensure that all the signals are detectable, we will need an additional assumption about the magnitude of those non-zero $\beta_i^\star$ values.  Specifically, consider 
\begin{equation}
\label{eq:betamin}
\min_{i \in S_{\beta^\star}} |\beta_i^\star| \geq ( M k_\alpha^{-1} \log n )^{1/2}, \quad \text{for some $M > 2$}, 
\end{equation}
where $k_\alpha = \frac{\alpha}{2(1+\alpha)}$.  Up to constants, condition \eqref{eq:betamin} is equivalent to the ``beta-min condition'' common in the high-dimensional estimation literature.  

\begin{thm}
\label{thm:selection}
Let $q^n$ be the variational approximation with $\lambda_n$ in \eqref{eq:lambda}.  Then 
\[ \E_{\beta^\star} q^n(\{S: S \supset S_{\beta^\star}\}) \to 0, \quad n \to \infty. \]
Moreover, if $\beta^\star$ is such that \eqref{eq:betamin} holds, then $\E_{\beta^\star} q^n(\{S: S \not\supseteq S_{\beta^\star}\}) \to 0$.  If all the above conditions hold, then the two conclusions can be combined, which implies that 
\[ \E_{\beta^\star} q^n(S_{\beta^\star}) \to 1, \quad n \to \infty. \]
\end{thm}

Next, one might be interested in the coverage probability of the credible sets derived from the variational approximation.  Such sets might include marginal credible intervals for individual $\beta_i$ or perhaps a linear combination of the full $\beta$ vector.  In the present context, we can prove that these marginal credible intervals achieve the target frequentist coverage probability asymptotically.  More precisely, let $w \in \RR^n$ be some fixed vector and define the linear functional $\omega = w^\top \beta$ of the full $\beta$ vector.  If $w$ is a standard basis vector, then the results below can be used to derive valid credible sets for an individual entry $\beta_i$ which, in turn, could be applied to all the entries to obtain a componentwise credible band.  Similarly, $w$ could consist of certain contrasts.  Regardless, the corresponding marginal posterior distribution for $\omega$ under the variational approximation $q^n$, which we denote by $q_\omega^n$, is given by 
\[ q_\omega^n(A) = \sum_S q^n(S) \nm(A \mid \hat\omega_S, \sigma^2 v_\alpha \|w_S\|^2), \quad A \subseteq \RR, \]
where $\hat\omega_S = w_S^\top y_S$ and $v_\alpha = (\alpha + \gamma)^{-1}$.  Intuitively, since $q^n(S_{\beta^\star}) \to 1$ according to  Theorem~\ref{thm:selection}, we expect 
\[ q_\omega^n(A) \approx q_\omega^{n,\oracle}(A) := \nm(A \mid \hat\omega_{S^\star}, \sigma^2 v_\alpha \|w_{S^\star}\|^2), \quad \text{all large $n$}. \]
As the following theorem demonstrates, this intuition is correct.  Moreover, the above approximation is sufficiently strong that credible intervals based on $q_\omega^n$ on the left-hand side are approximately credible intervals for the normal distribution on the right-hand side above.  And since the latter are known to be valid confidence intervals, the former must be so too, at least approximately.  

Without loss of generality, for $\zeta \in (0,\frac12)$, we consider $100(1-\zeta)$\% credible upper bounds for $\omega$ of the form 
\[ (-\infty, \bar\omega_\zeta] \quad \text{and} \quad (-\infty, \bar\omega_\zeta^\oracle] \]
based on $q_\omega^n$ and the oracle normal posterior $q_\omega^{n,\oracle}$, respectively.  These are simply upper $\zeta$-quantiles of these two posterior distributions, which (implicitly) depend on data through the posteriors.  The claim is that the coverage probability of the latter $100(1-\zeta)$\% credible upper bound is approximately equal to $1-\zeta$, i.e., 
\begin{equation}
\label{eq:interval.coverage}
\prob_{\beta^\star}(\bar\omega_\zeta < w^\top \beta^\star) \leq \zeta + o(1), \quad \text{as $n \to \infty$}. 
\end{equation}
Of course, if $v_\alpha \geq 1$ or, equivalently, if $\alpha + \gamma \leq 1$, then the former posterior's credible bounds have exact coverage probability in the sense that $\prob_{\beta^\star}(\bar\omega_\zeta^\oracle < w^\top \beta^\star) \leq \zeta$ for all $n$.  

\begin{thm}
\label{thm:uq.interval}
Under the conditions of Theorem~\ref{thm:selection}, if $\alpha + \gamma \leq 1$, then the coverage probability of $100(1-\zeta)$\% approximate posterior credible upper bound for $\omega$ under $q_\omega^n$ is approximately $1-\zeta$ as $n \to \infty$.
\end{thm}

\subsection{Numerical illustrations}
\label{SS:examples.orthogonal}

In this section, we consider
the normal means model described in Section~\ref{orthogonal}. For this model, we are interested at comparing VB-Laplace and VB-Empirical since these two methods have strong theoretical support. For simplification, we compare VB-Laplace and VB-empirical with known $\lambda_n$ as in \eqref{eq:lambda}, that is, neither method has hyperparameters to be updated.  We explore six different settings as follows, where the first three cases have large signals and the last three cases have small signals.
\begin{enumerate}
    \item $n=500,\; s=50,\; \beta_i=10$.
    \item $n=1000,\;s=100,\;\beta_i=10$.
    \item $n=2000,\;s=200,\;\beta_i=10$.
    \item $n=500,\; s=50,\; \beta_i=2$.
    \item $n=1000,\;s=100,\;\beta_i=2$.
    \item $n=2000,\;s=200\;,\beta_i=2$. 
\end{enumerate}
For each case, we compare the averaged $\ell_2$ estimation error and averaged credible interval length for important ones. The simulation results are recorded in Table~\ref{table:identity}. We further plot the averaged coverage probabilities for the first 20\% variables in each case in Figure~\ref{fig:identity}. In all cases, VB-empirical have smaller $l_2$ estimation error and higher averaged coverage probabilities than VB-Laplace. For the first three big signal cases, VB-Laplace and VB-empirical have nearly the same averaged credible region length for those variables identified as important 
ones but differ a lot in mean estimations, which is the main reason for the differences regarding inclusion probabilities. For the last three small signal cases, these two variational methods have similar mean estimation error but VB-empirical have larger credible region length than VB-Laplace, and the wider interval length explains the reason that VB-empirical still has higher coverage probabilities for small signal cases. Also, in this normal means case, VB-empirical only need one update while VB-Laplace still require a few iterations to get convergence, hence, VB-empirical is more efficient in this setting than VB-Laplace.  

\begin{table}[t]
\begin{center}
\begin{tabular}{c c c c c c}
 \hline
 Case & Method&$\E\|\hat\beta-\beta^\star\|^2$ (SE) & Mean length \\
 \hline
  1&VB-Laplace&10.00(0.90)&3.92(0.001)\\
    &VB-Emp&7.04(0.66)&3.92(0.001)\\
\hline
 2&VB-Laplace&14.10(0.90)&3.92(0.002)\\
    &VB-Emp&9.93(0.74)&3.92(0.001)\\
 \hline
 3&VB-Laplace&19.96(1.00)&3.92(0.002)\\
    &VB-Emp&14.08(0.71)&3.92(0.001)\\
     \hline
 4&VB-Laplace&13.87(0.13)&2.57(0.66)\\
    &VB-Emp&13.78(0.18)&3.57(0.66)\\
     \hline
 5&VB-Laplace&19.76(0.09)&3.48(0.43)\\
    &VB-Emp&19.64(0.15)&4.34(0.43)\\
     \hline
 6&VB-Laplace&28.07(0.07)&3.51(0.70)\\
    &VB-Emp&27.94(0.12)&4.22(0.70)\\
 \hline
\end{tabular}
\end{center}
\caption{Results for Simulation IV}
\label{table:identity}
\end{table}

\begin{figure}[t]
\begin{center}
\subfigure[Case 1]{\includegraphics[width=0.31\textwidth]{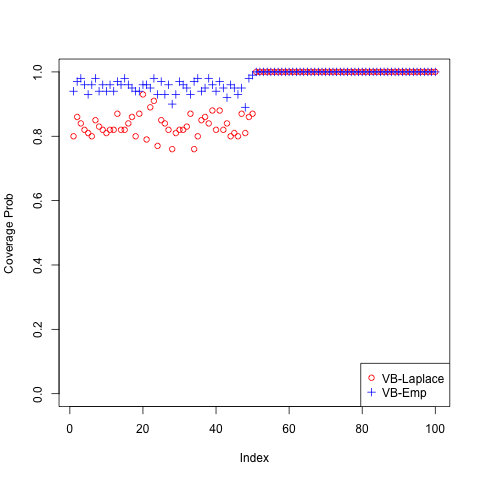}}
\subfigure[Case 2]{\includegraphics[width=0.31\textwidth]{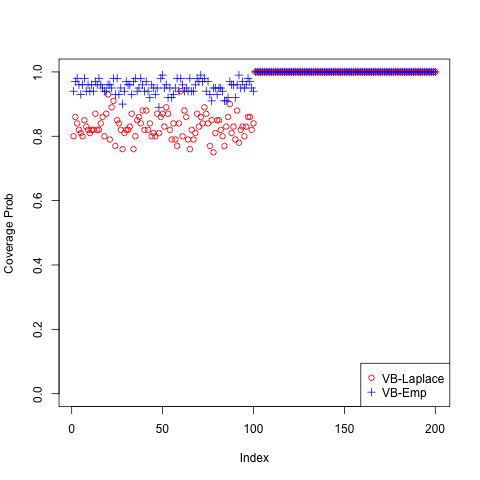}}
\subfigure[Case 3]{\includegraphics[width=0.31\textwidth]{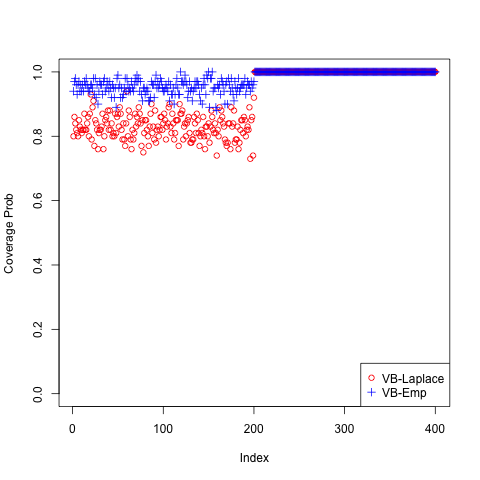}}
\subfigure[Case 4]{\includegraphics[width=0.31\textwidth]{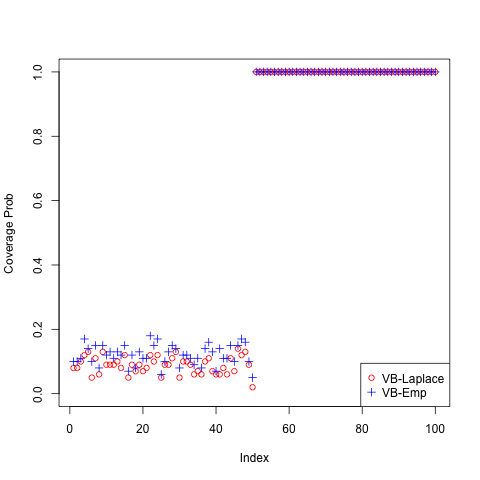}}
\subfigure[Case 5]{\includegraphics[width=0.31\textwidth]{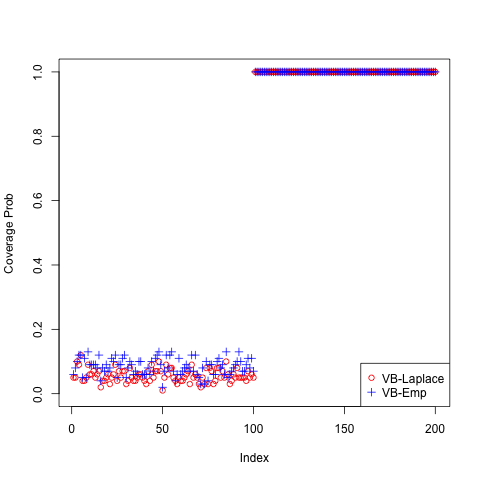}}
\subfigure[Case 6]{\includegraphics[width=0.31\textwidth]{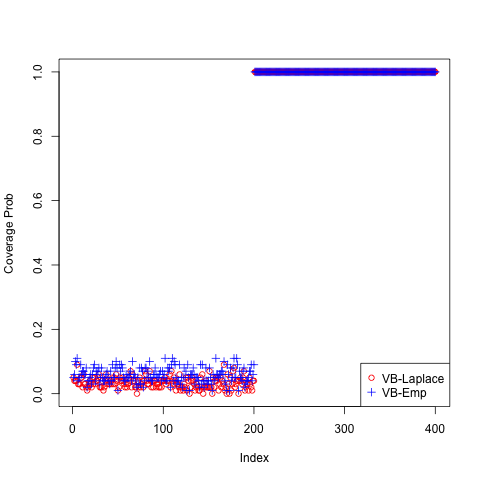}}
\end{center}
\caption{Averaged coverage probability for the first 20\% variables}
\label{fig:identity}
\end{figure}

\section{Conclusion} 
\label{S:discuss}

In this paper, we have proposed a variational approximation to the posterior distribution resulting from the empirical prior framework in \citet{martin2017empirical}.  In particular, this approximate posterior ignores the correlation between the entries of the vector $\beta$ inherent in the full posterior and, in return, can be computed very fast.  Aside from computational efficiency, our theoretical results show that this approximation makes no sacrifices in terms of asymptotic concentration rates compared to the full posterior.  Finally, we develop a coordinate ascent-based algorithm that borrows certain features from the importance sampling strategy in \citet{carbonetto2012scalable} and has superior empirical performance compared to other variational approximation methods across a range of different simulation scenarios.  

An advantage of the empirical prior formulation is that the posterior computations are generally faster/easier than their fixed prior counterparts.  This is because the data-driven center makes the prior tails less crucial to the posterior concentration properties, hence simple conjugate forms can be used.  While this conjugacy is helpful for the structure-specific parameters, high-dimensional problems like the one considered here have an unknown structure---e.g., the configuration $S$---and posterior sampling of the structure will require MCMC.  Since MCMC is relatively slow and cannot be completely avoided when working with the genuine posterior, there remains an interest in finding fast approximations, and this is what motivated our efforts here.  However, there are other examples involving structured, high-dimensional parameters where computational efficiency can be gained by working with variational approximations.  For example, \citet{eb.gwishart} and \citet{ebpiecep} develop empirical priors and posterior concentration rate results for sparse, high-dimensional precision matrix and piecewise polynomial signal estimation, respectively, and both could benefit from accelerated posterior computations via variational approximations.  We leave this as a topic for future work.

\appendix

\section{Derivation of the update equations}
\label{derivation}

In this section, we present our derivation details for update equations in  \eqref{eqn:update_equ}.  The relevant quantity is the following log-ratio:
\begin{align*}
\log \frac{\Tilde{\pi}^n(S,\beta)}{q_{\theta}(S,\beta)} & = -\tfrac{\alpha}{2\sigma^2}\bigl( \beta_S^\top X_S^\top X_S \beta_S - 2\beta_S^\top X_S^\top y \bigr) \\
& \qquad - \tfrac{\gamma}{2\sigma^2} (\beta_S - \hat\beta_S)^\top X_S^\top X_S (\beta_S - \hat\beta_S) - \textstyle\sum_{j=1}^p (1-S_j )\log(1-\phi_j) \\
& \qquad - \textstyle\sum_{j=1}^p S_j \bigl\{ \log \phi_j - \tfrac12\log 2\pi - \tfrac12\log \tau_j^2 - \tfrac{1}{2\tau_j^2}(\beta_j - \mu_j)^2 \bigr\} - \log\textstyle \binom{p}{|S|} \\
& \qquad + \tfrac12 \log|\gamma X_S^\top X_S| - |S|(\tfrac12\log\sigma^2 + \tfrac12\log 2\pi + \log c + a\log p) 
\end{align*}
where $\theta$ consists of $\{(\mu_j, \tau_j^2, \phi_j): j=1,\ldots,p\}$, $\hat\beta_S$ is the least squares estimator under model $S$, and, with a slight abuse of notation, $S$ denotes both a subset of $\{1,\ldots,p\}$ and a vector $(S_1,\ldots,S_p)$ of binary variables; we have also ignored constants (in $\theta$) in the right-hand side above as these do not affect the optimization.  The next step is to take expectation of the above expression with respect to $(S,\beta) \sim q_\theta$, i.e., with respect to the variational approximation.  Some of the terms will be easy to deal with while others are more involved.  It will help to give those challenging terms names:
\begin{align*}
A & = \log\textstyle\binom{p}{|S|} \\
B & = (\beta_S - \hat\beta_S)^\top X_S^\top X_S (\beta_S - \hat\beta_S) \\
C & = \log |\gamma X_S^\top X_S| \\
D & = \beta_S^\top X_S^\top X_S \beta_S - 2\beta_S^\top X_S^\top y.
\end{align*}
We will evaluate/approximate the expected value of each in turn.  
\begin{itemize}
\item Note that $\E(A)$ under the variational approximation will be a non-linear function of the vector $\phi$, which would be non-trivial to differentiate, etc., so we seek a simpler approximation.  For any integer $s \leq p$, recall that 
\[ s \log(p / s) \leq \log\textstyle\binom{p}{s} \leq s + s \log(p / s) \]
If $s \ll p$, then both the lower and upper bounds are relatively large, but mostly insensitive to small changes in $s$.  Under the posterior $\pi^n$ being approximated, we expect $|S| \ll p$ and, therefore, $\E(A) = \E\{\log\binom{p}{|S|}\}$ would not be sensitive to small changes in $\phi$.  Consequently, the gradient of $\E(A)$ would be small, so we opt to simply ignore this term when solving the optimization problem.   
\item Next, for $\E(B)$ under the variational approximation, note that the conditional distribution of $\beta_S$, given $S$, is $\nm_{|S|}(\mu_S, D_S)$, where $D_S = \text{diag}(\tau_S^2)$.  By iterated expectation and the familiar expression for expected value of quadratic forms, we get 
\begin{align*}
\E(B) & = \E\{ \trace(X_S^\top X_S D_S) + (\mu_S - \hat\beta_S)^\top X_S^\top X_S (\mu_S - \hat\beta_S) \} \\
& = n \sum_{j=1}^p \tau_j^2 \E(S_j) + \E\{(\mu_S - \hat\beta_S)^\top X_S^\top X_S (\mu_S - \hat\beta_S)\}, 
\end{align*}
where the remaining expectation is with respect to the marginal (variational) distribution of $S$.  Of course, $\E(S_j) = \phi_j$, but the second term requires some more work. As a first step, recall that 
\[ \lambda_{\min}(S) \|\mu_S - \hat\beta_S\|^2 \leq (\mu_S - \hat\beta_S)^\top X_S^\top X_S (\mu_S - \hat\beta_S) \leq \lambda_{\max}(S) \|\mu_S - \hat\beta_S\|^2, \]
where $\lambda_{\min}(S)$ and $\lambda_{\max}(S)$ are the minimum and maximum eigenvalues of $X_S^\top X_S$, respectively.  The lower and upper bounds are extremes, and it is not unreasonable to approximate the quadratic form on the inside by an ``average'' value of the two extremes.  In particular, we take 
\[ (\mu_S - \hat\beta_S)^\top X_S^\top X_S (\mu_S - \hat\beta_S) \approx g(S) \|\mu_S - \hat\beta_S\|^2, \]
where $g(S)$ is the geometric mean of the eigenvalues of $X_S^\top X_S$.  Moreover, as discussed in Section~\ref{SS:algorithm}, it is also not unreasonable to assume that $g(S)$ is relatively stable in $S$, which suggests a further approximation  
\[ (\mu_S - \hat\beta_S)^\top X_S^\top X_S (\mu_S - \hat\beta_S) \approx g(\tilde S) \|\mu_S - \hat\beta_S\|^2,\]
where $\tilde S$ is the configuration selected by the lasso estimator $\tilde\beta$.  We propose to plug in $\tilde\beta_S$, the sub-vector of the lasso estimator $\tilde\beta$ corresponding to configuration $S$; this is computationally than working with $\hat\beta_S$ because the latter is not a sub-vector, it requires a separate calculation, $(X_S^\top X_S)^{-1} X_S^\top y$, for each $S$.  
Finally we have
\begin{equation}
\label{eq:B}
\E(B) \approx \sum_{j=1}^p  n\tau_j^2 \phi_j +g(\tilde S) \sum_{i=1}^p \phi_j (\mu_j - \tilde \beta_j)^2. 
\end{equation}

\item For $C$, note first that we have $\log|\gamma X_S^\top X_S| = |S| \log\{\gamma g(S)\}$, where $g(S)$ is the geometric mean as described above.  From our discussion above, based on our full rank assumption in \eqref{eq:fullrank}, our standardization strategy in \eqref{eq:standardize}, and the remarks in Section~\ref{SS:algorithm}, 
it is not unreasonable to assume that $g(S)$ is relatively stable across $S$, so we end up with the approximation 
\begin{equation}
\label{eq:C}
\E(C) \approx \log\{\gamma g(\tilde S)\} \sum_{j=1}^p \phi_j, 
\end{equation}
where $\tilde S$ is, e.g., the configuration chosen by lasso.  
\item Finally, for $D$, again by iterated expectation and the general expression for expected values of quadratic forms, we have 
\[ \E(D) = \E(\trace(X_S^\top X_S D_S)) + \E(\mu_S^\top X_S^\top X_S \mu_S) - 2y^\top \E(X_S \mu_S), \]
where the expectation on the right-hand side is with respect to $S$ under the variational approximation.  Similar calculations as with $B$ above, yield 
\begin{equation}
\label{eq:D}
\E(D) = n\sum_{j=1}^p \phi_j (\tau_j^2+\mu_i^2) + \sum_{j=1}^p \sum_{k\neq j}^p \phi_j \phi_k (X^\top X)_{jk} \mu_j \mu_k - 2 \sum_{i=1}^n y_i \sum_{j=1}^p x_{ij} \phi_j \mu_j. 
\end{equation}
\end{itemize}
Putting everything together, we find that the function to be optimized is 
\begin{align*}
\E\log \frac{\Tilde{\pi}^n(S,\beta)}{q_{\theta}(S,\beta)} & \approx -\frac{\alpha \E(D)}{2\sigma^2} - \frac{\gamma \E(B)}{2\sigma^2} - \sum_{j=1}^p (1-\phi_j)\log(1-\phi_j) \\
& \qquad - \sum_{j=1}^p \phi_j \bigl\{ \log \phi_j - \tfrac12\log 2\pi - \tfrac12\log \tau_j^2 - \tfrac{1}{2} \bigr\} + \frac{\E(C)}{2} \\
& \qquad - (\log\sigma^2 + \log 2\pi + \log c + a\log p) \sum_{j=1}^p \phi_j
\end{align*}
Plugging \eqref{eq:B}, \eqref{eq:C}, and \eqref{eq:D} into the above expression, we get 

\begin{align*}
 \sum_{j=1}^p \Big[&-\frac{\alpha}{2\sigma^2}\Bigl\{ n \phi_j (\tau_j^2+\mu_j^2) +  \phi_j \mu_j\sum_{k\neq j}^p (X^\top X)_{jk} \phi_k \mu_k - 2 \phi_j \mu_j \sum_{i=1}^n x_{ij}y_i\Bigr\} \\
& -\frac{\gamma}{2\sigma^2}\bigl\{ n\tau_j^2 \phi_j +g_{\tilde S} \phi_j (\mu_j - \tilde \beta_j)^2 \bigr\}-(1-\phi_j)\log(1-\phi_j)\\
&-\phi_j \bigl( \log \phi_j - \tfrac12\log 2\pi - \tfrac12\log \tau_j^2 - \tfrac{1}{2} \bigr)+\frac{\log(\gamma g_{\tilde{S}})}{2}\phi_j\\
&-\phi_j (\tfrac12\log\sigma^2 + \tfrac12\log 2\pi + \log c + a\log p)\Big]
\end{align*}

We see that there is an overall sum over $j=1,\ldots,p$, and differentiating with respect to each of $\phi_j$, $\mu_j$, and $\tau_j$ leads to the updates equations in \eqref{eqn:update_equ}.

\section{Background theory}
\label{summarytheory}

In this section, we present a summary of some results that will be used in the proofs of our main theorems.  The driving result is Theorem~7 in \citet{ray2019variational}, which shows how the probabilities of certain events under the variational approximation, $q^n$, can be controlled in terms of the probability of the same event under the true posterior, $\pi^n$, and the Kullback--Leibler divergence of the latter from the former.  

\begin{prop}
\label{prop:Th7}
Let $B_n \subset \RR^p$ be a sequence of events about $\beta$.  If, 
\[ \E_{\beta^\star} \pi^n(B_n) \leq C e^{-\delta_n}, \]
for some constant $C > 0$ and sequence $\delta_n > 0$, then 
\[ \E_{\beta^\star} q^n(B_n) \leq 2 \delta_n^{-1} \bigl\{ \E_{\beta^\star} K(q^n, \pi^n) + C e^{-\delta_n / 2} \bigr\}, \]
where $K(q^n, \pi^n)$ is the Kullback--Leibler divergence of $\pi^n$ from $q^n$. 
\end{prop}

Since the properties of interest here concern probabilities assigned by $q^n$ to certain events, Proposition~\ref{prop:Th7} provides us with a strategy to prove these claims: first, establish exponential inequalities for the $\pi^n$-probability assigned to those relevant events and, second, bound the Kullback--Leibler divergence.  Fortunately, the aforementioned exponential inequalities have been established elsewhere; see \citet{martin2017empirical} and \citet{martin2019empiricalpredict}.  Below is a summary. 

Briefly, let $\pi^n$ denote the posterior distribution based on the empirical prior as described in Section~\ref{SS:posterior} above, with $\sigma^2$ taken to be known.

\begin{prop}
\label{prop:post_pred}
For $\eps_n^2$ in \eqref{eq:rate} and any sequence $M_n$, with $M_n \to \infty$, the posterior distribution $\pi^n$ satisfies 
\[ \E_{\beta^\star} \pi^n(\{\beta \in \RR^p: \|X(\beta-\beta^\star)\|_2^2 > M_n \eps_n^2(\beta^\star)\}) \lesssim e^{-M_n' \eps_n^2(\beta^\star)}, \]
for all $\beta^\star$ with $|S_{\beta^\star}|=o(n)$, where $M_n'$ is another sequence with $M_n' \sim M_n$.  
\end{prop}

\begin{prop}
\label{prop:effective_dimension}
Under the setup of Proposition~\ref{prop:post_pred}, for any sequence $M_n > 1$ with $M_n \to \infty$, the posterior $\pi^n$ satisfies 
\[ \E_{\beta^\star} \pi^n(\{\beta \in \RR^p: |S_\beta| > M_n |S_{\beta^\star}|\}) \lesssim e^{-M_n' |S_{\beta^\star}| \log n}, \]
for all $\beta^\star$ with $|S_{\beta^\star}|=o(n)$, where $M_n'$ is another sequence with $M_n' \sim M_n$.  
\end{prop}

\begin{prop}
\label{prop:post_est}
Under the setup in Proposition~\ref{prop:post_pred}, for any sequence $M_n$ such that $M_n \to \infty$, the posterior $\pi^n$ satisfies 
\[ \E_{\beta^\star} \pi^n\Bigl( \Bigl\{\beta\in \RR^p: \|\beta-\beta^\star\|^2_2 > \frac{M_n \eps_n^2(\beta^\star)}{\kappa_X^2(C|S_{\beta^\star}|)} \Bigr\}\Bigr) \lesssim e^{-M_n' \eps_n^2(\beta^\star)},\]
for all $\beta^\star$ such that $|S_{\beta^\star}|=o(n)$ and $\kappa_X(C|S_{\beta^\star}|) > 0$, for $\kappa_X$ in \eqref{eq:kappa} and a given constant $C > 2$, where $M_n'$ is another sequence with $M_n' \sim M_n$.  
\end{prop}

\section{Proofs from Section~\ref{theory}}
\label{App:proof}

The essential step in the proofs of our three main results in Section~\ref{theory} is to bound the Kullback--Leibler divergence over a subset of the mean-field family $\Qset$.  Let $\Qset'$ denote the collection of all $q \in \Qset$ but with mixture weights $\phi_j$ that are either 0 or 1.  That is, $\Qset'$ consists of distributions that are products of normals and point masses---no mixtures.  And since $\Qset' \subset \Qset$, we have 
\[ \min_{q \in \Qset} K(q, \pi^n) \leq \min_{q \in \Qset'} K(q, \pi^n). \]
The advantage is that $\Qset'$ consists of simpler distributions so bounding the Kullback--Leibler divergence over $\Qset'$ is an easier task.  In fact, the minimum Kullback--Leibler divergence over $\Qset'$ is smaller than if we fix $\phi$ so that $\phi_j = 1$ for $j \in S^\star$ and $\phi_j = 0$ for $j \not\in S^\star$, where $S^\star=S_{\beta^\star}$.  Therefore, 
\begin{align*}
\min_{q \in \Qset'} & \, K(q, \pi^n) \\
& \leq \int \log \frac{\nm_{|S^\star|}(d\beta_{S^\star} \mid \mu_{S^\star}, \tau_{S^\star}^2) \otimes \delta_0(d\beta_{S^{\star c}})}{\sum_S \pi^n(S) \nm_{|S|}(d\beta_S \mid \hat\beta_S, V_{S^\star}) \otimes \delta_0(d\beta_{S^c})} \nm_{|S^\star|}(d\beta_{S^\star }\mid \mu_{S^\star}, \tau_{S^\star}^2) \otimes \delta_0(d\beta_{S^{\star c}}) \\
& = -\log \pi^n(S^\star) + K(q_{S^\star}, \pi_{S^\star}^n), 
\end{align*}
where $V_{S^\star} = \sigma^2(\alpha + \gamma)^{-1} (X_{S^\star}^\top X_{S^\star})^{-1}$, and $q_{S^\star}$ and $\pi_{S^\star}^n$ are the corresponding conditional distributions of $\beta_{S^\star}$, given $S=S^\star$.  A closed-form expression is available for the first term in the upper bound and, since both $q_{S^\star}$ and $\pi_{S^\star}^n$ are Gaussian, the second term can be evaluated too.  In what follows, we bound each of these two terms in turn. 

For the marginal $\pi^n$-probability at $S^\star$, recall that 
\[ \pi^n(S^\star) = \frac{\pi(S^\star) \, \bigl(\tfrac{\gamma}{\alpha+\gamma}\bigr)^{|S^\star|/2} \exp\bigl\{-\tfrac{\alpha}{2\sigma^2}\|y-\hat y_{S^\star}\|^2\bigr\}}{\sum_S \pi(S) \, \bigl(\tfrac{\gamma}{\alpha+\gamma}\bigr)^{|S|/2} \exp\bigl\{-\tfrac{\alpha}{2\sigma^2}\|y-\hat y_{S}\|^2\bigr\}}. \]
Since $\|y-\hat y_{S^\star}\|^2 \leq \|y - X\beta^\star\|^2$, we get 
\[ \pi^n(S^\star) \geq D_n^{-1} \pi(S^\star) \, \bigl(\tfrac{\gamma}{\alpha+\gamma}\bigr)^{|S^\star|/2}, \]
where 
\begin{align*}
D_n & = \sum_S \pi(S) \, \bigl(\tfrac{\gamma}{\alpha+\gamma}\bigr)^{|S|/2} \exp\bigl[\tfrac{\alpha}{2\sigma^2}\{\|y-X\beta_{S+}\|^2 - \|y-\hat y_{S}\|^2\} \bigr] \\
& = \sum_S \pi(S) \int R_n^\alpha(\beta_{S+}) \, \pi_n(\beta_S \mid S) \, d\beta_S, 
\end{align*}
with $R_n(\beta_{S+}) = L_n(\beta_{S+})/L_n(\beta^\star)$ and $\beta_{S+} = (\beta_S, 0_{S^c})$, the $p$-vector with zeros filled in around the non-zero $\beta_S$.  This quantity $D_n$ is precisely the denominator of the posterior distribution $\pi^n$ that appears in Lemma~1 of \citet{martin2017empirical}.  Taking negative logarithm and then expectation, gives
\[ \E_{\beta^\star} \{-\log \pi^n(S^\star)\} \leq -\log \pi(S^\star) + |S^\star| \bigl(\tfrac12 \log\tfrac{\alpha + \gamma}{\gamma}\bigr) + \E_{\beta^\star} \log D_n. \]
By Jensen's inequality, 
\[ \E_{\beta^\star} \log D_n \leq \log \E_{\beta^\star} D_n = \sum_S \pi(S) \int \E_{\beta^\star} \{R_n^\alpha(\beta_{S+}) \, \pi_n(\beta_S \mid S)\} \, d\beta_S. \]
The same H\"older's inequality argument in Lemma~2 of \citet{martin2017empirical} can be used to show that the integral on the right-hand side above equals $\psi^{|S|}$, where $\psi$ is a constant that depends on $(\alpha, \gamma, \sigma^2)$ only.  Since the prior for $|S|$ has very thin tails, the sum on the right-hand side is uniformly bounded.  Therefore, the dominant term in the upper bound for $\E_{\beta^\star}\{-\log \pi^n(S^\star)\}$ is $-\log \pi(S^\star)$, and since 
\[ \log \binom{p}{s} \leq s \log(ep/s), \]
we get 
\[ \E_{\beta^\star} \{-\log \pi^n(S^\star)\} \lesssim |S^\star| \log(p / |S^\star|) + a |S^\star| \log p. \]
Since 
\[ |S^\star| \log p = \Bigl( 1 + \frac{\log |S^\star|}{\log(p / |S^\star|)} \Bigr) |S^\star| \log(p / |S^\star|), \]
and the ratio inside the parentheses is bounded, we get $\E_{\beta^\star}\{-\log \pi^n(S^\star)\} \lesssim \eps_n^2(\beta^\star)$.  

Next, the second term, $K(q_{S^\star}, \pi_{S^\star}^n)$, in the above upper bound is the Kullback--Leibler divergence between two $|S^\star|$-variate Gaussians, so a direct calculation is possible.  For $q_{S^\star}$, we are free to choose the mean vector and diagonal covariance matrix as we please and, of course, we set the mean vector equal to that of $\pi_{S^\star}^n$, so we get 
\[ K(q_{S^\star}, \pi_{S^\star}^n) = \tfrac12 \bigl\{ \log|V_{S^\star} \Delta_{S^\star}^{-1}| - |S^\star| + \trace(V_{S^\star}^{-1} \Delta_{S^\star}) \bigr\}, \]
where $\Delta_{S^\star}$ is a diagonal matrix.  If we set 
\[ \Delta_{S^\star} = \{\text{diag}(V_{S^\star}^{-1})\}^{-1}, \]
then $\trace(V_{S^\star}^{-1} \Delta_{S^\star}) = |S^\star|$ and 
\[ |\Delta_{S^\star}^{-1}| = \prod_{j=1}^{|S^\star|} (V_{S^\star}^{-1})_{jj} \leq \Bigl( \frac{\alpha + \gamma}{\sigma^2} \Bigr)^{|S^\star|} \prod_{j=1}^{|S^\star|} (X_{S^\star}^\top X_{S^\star})_{jj} = \Bigl\{ \frac{n(\alpha + \gamma)}{\sigma^2} \Bigr\}^{|S^\star|}, \]
where the last equality follows as a result of how we standardized of the columns of $X$.  Finally, we have that 
\[ \log |V_{S^\star} \Delta_{S^\star}^{-1}| \leq |S^\star| \log\bigl\{ \lambda_{\min}^{-1} (n^{-1} X_{S^\star}^\top X_{S^\star}) \bigr\}. \]
According to \eqref{eq:eigen}, the eigenvalue is lower bounded by $|S_{\beta^\star}| p^{-1}$ and, therefore, the right-hand side is upper bounded by $\eps_n^2(\beta^\star)$.

To summarize, we have a bound of order $\eps_n^2 = \eps_n^2(\beta^\star)$ on $K(q^n, \pi^n)$ and the exponential inequalities in Propositions~\ref{prop:post_pred}--\ref{prop:post_est}.  Applying the result in Theorem~7 of \citet{ray2019variational}, presented above in Proposition~\ref{prop:Th7}, the proofs of the three main theorems follow directly.  For example, according to Propositions~\ref{prop:Th7} and \ref{prop:post_est},
\begin{align*}
\E_{\beta^*} q^n(\{\beta: \|X(\beta-\beta^\star)\|_2^2>M_n\eps_n^2\})&\lesssim \frac{\E_{\beta^\star}K(q^n,\pi^n)+e^{-M_n\eps_n^2/2}}{M_n\eps_n^2}\\
& \lesssim  (M_n\eps_n^2)^{-1}\{ \eps_n^2+e^{-M_n\eps_n^2/2}\}\to 0.
\end{align*}
This proves Theorem~\ref{prediction_error}; Theorems~\ref{effective dimension} and \ref{estimation_error} can be proved similarly.

\section{Proofs from Section~\ref{SS:theory.orthogonal}}

\begin{proof}[Proof of Theorem~\ref{thm:selection}]
If $q^n$ is the variational approximation, for any $S$ we can write 
\[ q^n(S) = \prod_{i \in S} \phi_i \, \prod_{i \not\in S} (1-\phi_i), \]
where $\phi_1,\ldots,\phi_n$ are the weights given by 
\[ \logit(\phi_i) = \logit(\lambda_n) - \log z + \tfrac{\alpha}{2} y_i^2, \quad i=1,\ldots,n, \]
determined by minimizing the Kullback--Leibler divergence, where $\lambda_n$ is as in \eqref{eq:lambda} and $z=(1+\alpha \gamma^{-1})^{1/2}$.  Note that both $z$ and $\lambda_n$ do not depend on data.  If we write $S^\star = S_{\beta^\star}$, then we get the following convenient bound
\[ q^n(S) \leq \frac{q^n(S)}{q^n(S^\star)} = \prod_{i \in S \cap S^{\star c}} e^{\logit(\phi_i)} \, \prod_{i \in S^c \cap S^\star} e^{-\logit(\phi_i)}. \]
Since each $\phi_i$ only depends on $y_i$, and these are independent, we can interchange the order of expectation and product.  Also, for those $i \in S^{\star c}$, with $\beta_i^\star=0$, the $\phi_i$'s are iid, so each term in that product has the same expectation.  Therefore, 
\[ \E_{\beta^\star} q^n(S) \leq \bigl\{ \E_0 e^{\logit(\phi_1)} \bigr\}^{|S \cap S^{\star c}|} \, \prod_{i \in S^c \cap S^\star} \E_{\beta_i^\star} e^{-\logit(\phi_i)}. \]
Using the moment generating function formulas for the central and non-central chi-square distributions, it is easy to check that 
\begin{align*}
\E_0 e^{\logit(\phi_1)} & = \exp\{\logit(\lambda_n) - \log z - \tfrac12 \log (1-\alpha)\} \\
\E_{\beta_i^\star} e^{-\logit(\phi_i)} & = \exp\{-\logit(\lambda_n) + \log z - \tfrac12\log(1 + \alpha) -k_\alpha \beta_i^{\star 2} \}, 
\end{align*}
where $k_\alpha = \frac{\alpha}{2(1+\alpha)}$.  We consider two distinct cases separately, namely, $S \supset S^\star$ and $S \not\supseteq S^\star$.  First, for any $S \supset S^\star$, we have that $|S^c \cap S^\star|=0$.  So, 
\begin{align*}
\E_{\beta^\star} q^n(\{S: S \supset S^\star\}) & \leq \sum_{S: S \supset S^\star, |S| \leq C|S^\star|} \{ \E_0 e^{\logit(\phi_1)} \}^{|S \cap S^{\star c}|} \\
& = \sum_{t=1}^{(C-1)|S^\star|} \binom{n-|S^\star|}{t} \{ \E_0 e^{\logit(\phi_1)} \}^t \\
& \leq \sum_{t=1}^{(C-1)|S^\star|} \{ e(n - |S^\star|)\E_0 e^{\logit(\phi_1)} \}^t \\ 
& \lesssim n e^{\logit(\lambda_n)}.
\end{align*}
For $\lambda_n$ as in \eqref{eq:lambda}, the upper bound is vanishing as $n \to \infty$.  Next, for any $S \not\supseteq S^\star$, we know that there is at least one component in $S^\star$ that is {\em not included} in $S$.  So, if we set $\Delta = \min_{i \in S^\star} |\beta_i^\star|$, then we get 
\begin{align*}
\E_{\beta^\star} & \, q^n(\{S: S \not\supseteq S^\star\}) \\
& \leq \sum_{S: S \not\supseteq S^\star, |S| \leq C|S^\star|} \Bigl[ \bigl\{ \E_0 e^{\logit(\phi_1)} \bigr\}^{|S \cap S^{\star c}|} \, \prod_{i \in S^c \cap S^\star} \E_{\beta_i^\star} e^{-\logit(\phi_i)} \Bigr] \\
& \leq \sum_{S: S \not\supseteq S^\star, |S| \leq C|S^\star|} \{c_0 e^{\logit(\lambda_n)}\}^{|S \cap S^{\star c}|} \{c_1 e^{-\logit(\lambda_n)-k_\alpha \Delta^2}\}^{|S^c \cap S^\star|} \\
& = \sum_{s=0}^{C|S^\star|} \sum_{t=0}^{s \wedge (|S^\star|-1)} \binom{|S^\star|}{t} \binom{n-|S^\star|}{s-t} \{c_0 e^{\logit(\lambda_n)}\}^{s-t} \{c_1 e^{-\logit(\lambda_n)-k_\alpha \Delta^2}\}^{|S^\star|-t} \\ 
& \leq \sum_{s=0}^{C|S^\star|} \sum_{t=0}^{s \wedge (|S^\star|-1)} \{c_0 (n-|S^\star|)e^{\logit(\lambda_n)}\}^{s-t} \{c_1 |S^\star| e^{-\logit(\lambda_n)-k_\alpha \Delta^2}\}^{|S^\star|-t}. 
\end{align*}
(In the above derivation, $s$ represents $|S|$ and $t$ represents $|S \cap S^\star|$, which implies $s-t=|S \cap S^{\star c}|$ and $|S^\star|-t = |S^c \cap S^\star|$.)  Note that $t < |S^\star|$ because $S \not\supseteq S^\star$ implies that $S$ can't include all the entries in $S^\star$.  This means that there is a constant factor 
\[ |S^\star| e^{-\logit(\lambda_n)-k_\alpha \Delta^2}, \]
which goes to 0 as $n \to \infty$ if $\Delta$ is sufficiently large.  The terms involve 
\[ (n-|S^\star|) e^{\logit(\lambda_n)} \]
which also can be made to vanish as we discussed in the $S \supset S^\star$ case above.  So, with the exception of the common factor involving $\Delta$ above, all the terms are geometrically small and, hence, the sum is bounded.  Putting everything together, if the beta-min condition \eqref{eq:betamin} holds, then we can conclude that both $\E_{\beta^\star} q^n(\{S: S \supset S^\star\})$ and $\E_{\beta^\star} q^n(\{S: S \not\supseteq S^\star\})$ vanish, which proves the claim.  
\end{proof}

\begin{proof}[Proof of Theorem~\ref{thm:uq.interval}]
Define $D_n(A) = | q_\omega^n(A) - q_\omega^{n,\oracle}(A) |$ for Borel sets $A \subseteq \RR$.  Since $q_\omega^n$ is a mixture and $|\sum_i x_i| \leq \sum_i |x_i|$, we get the following upper bound:
\[ D_n(A) \leq \sum_S q^n(S) \bigl| \nm(A \mid \hat\omega_S, \sigma^2 v_\alpha \|w_S\|^2) - \nm(A \mid \hat\omega_{S^\star}, \sigma^2 v_\alpha \|w_{S^\star}\|^2) \bigr|. \]
The absolute difference is 0 when $S=S^\star$ and bounded by 2 otherwise, so the total variation distance between $q_\omega^n$ and $q_\omega^{n,\oracle}$ is upper bounded as follows:
\[ d_{\text{\sc tv}}\bigl( q^n_\omega, q_\omega^{n,\oracle} \bigr) := \sup_A |q_\omega^n(A) - q_\omega^{n,\oracle}(A)| \leq 2 \sum_{S \neq S^\star} q^n(S) = 2 \{1 - q^n(S^\star)\}. \]
Taking expectation of both sides and applying Theorem~\ref{thm:selection} gives 
\begin{equation}
\label{eq:tv}
\E_{\beta^\star} d_{\text{\sc tv}}( q^n_\omega, q_\omega^{n,\oracle} ) \to 0, \quad n \to \infty. 
\end{equation}
Towards the (non-)coverage result in \eqref{eq:interval.coverage}, recall that 
\[ \bar\omega_\zeta^\oracle = \hat\omega_{S^\star} + z_\zeta \sigma_\alpha \|w_{S^\star}\|, \]
where $\Phi(z_\zeta) = 1-\zeta$.  Also, $\hat\omega_{S^\star} \sim \nm(\omega^\star, \sigma^2 \|w_{S^\star}\|^2)$, where $\omega^\star = w^\top \beta^\star$.  Take any $t > 0$ and define $u(t) = \sigma \|w_{S^\star}\| t$.  Then we have  
\begin{align*}
\prob_{\beta^\star}(\bar\omega_\zeta < \omega^\star) & = \prob_{\beta^\star}\{\bar\omega_\zeta < \omega^\star, \bar\omega_\zeta^\oracle < \omega^\star + u(t)\} + \prob_{\beta^\star}\{\bar\omega_\zeta < \omega^\star, \bar\omega_\zeta^\oracle \geq \omega^\star + u(t)\} \\
& \leq \prob_{\beta^\star}\{\bar\omega_\zeta^\oracle < \omega^\star + u(t)\} + \prob_{\beta^\star}\{\bar\omega_\zeta < \omega^\star, \bar\omega_\zeta^\oracle \geq \omega^\star + u(t)\}. 
\end{align*}
Using the normal sampling distribution of $\hat\omega_{S^\star}$, we can easily see that the first term in the upper bound is $\Phi(t - v_\alpha^{1/2} z_\zeta)$.  Note that this is no more than $\zeta$ as $t \to 0$.  For the second term, note that the event in question determines a gap between the quantiles of $q_\omega^n$ and $q_\omega^{n,\oracle}$, which suggests a non-zero total variation distance.  Indeed, we know that 
\[ q_\omega^n((-\infty,\bar\omega_\zeta]) = 1-\zeta, \]
so if $\bar\omega_\zeta < \omega^\star$ and $\bar\omega_\zeta^\oracle \geq \omega^\star + u(t)$, then the $q_\omega^{n,\oracle}$-probability of that same interval satisfies 
\[ q_\omega^{n,\oracle}((-\infty, \bar\omega_\zeta]) \leq q_\omega^{n,\oracle}((-\infty,\omega^\star]) = (1-\zeta) - Q_\omega^{n,\oracle}([\omega^\star, \bar\omega_\zeta^\oracle]). \]
Since $\bar\omega_\zeta^\oracle$ is an upper quantile of $q_\omega^{n,\oracle}$ and the length of the interval $[\omega^\star, \bar\omega_\zeta^\oracle]$ is at least $u(t)$, the probability can be lower-bounded by 
\[ q_\omega^{n,\oracle}([\omega^\star, \bar\omega_\zeta^\oracle]) \geq q_\omega^{n,\oracle}([\hat\omega_{S^\star} + z_\zeta \sigma v_{\alpha}^{1/2}\|w_{S^\star}\| - u(t), \hat\omega_{S^\star} + z_\zeta \sigma v_{\alpha}^{1/2}\|w_{S^\star}\|]). \]
After standardizing, this probability is $\Phi(z_\zeta) - \Phi(z_\zeta - v_\alpha^{-1/2}t)$, so
\[ q_\omega^{n,\oracle}((-\infty, \bar\omega_\zeta]) \leq \Phi(z_\zeta - v_\alpha^{-1/2}t) < 1-\zeta. \]
This difference in probabilities implies a difference in total variation distance, i.e., 
\[ \prob_{\beta^\star}\{\bar\omega_\zeta < \omega^\star, \bar\omega_\zeta^\oracle \geq \omega^\star + u(t)\} \leq \prob_{\beta^\star}\{d_{\text{\sc tv}}( q^n_\omega, q_\omega^{n,\oracle} ) \geq (1-\zeta) - \Phi(z_\zeta - v_\alpha^{-1/2}t)\}. \]
Markov's inequality and \eqref{eq:tv} imply that the upper bound above vanishes.  Putting everything together, we have that 
\[ \limsup_{n \to \infty} \prob_{\beta^\star}(\bar\omega_\zeta < \xi^\star) \leq \Phi(t - v_\alpha^{1/2} z_\zeta) \quad \text{for any $t > 0$}. \]
But if the above inequality holds for all $t > 0$, then it must also hold for the infimum, and the right-hand side minimum value is $\Phi(-v_\alpha^{1/2} z_\zeta) \leq \zeta$, proving \eqref{eq:interval.coverage}. 
\end{proof}

\bibliography{reference}
\bibliographystyle{apalike}

\end{document}